\documentclass[11pt,a4paper]{article}

\usepackage[T1]{fontenc}
\usepackage[utf8]{inputenc}
\usepackage{microtype}
\usepackage{geometry}
\usepackage{amsmath,amsthm,mathtools,bm}
\usepackage{newtxtext,newtxmath}
\usepackage{booktabs,threeparttable,array,multirow}
\usepackage{graphicx}
\usepackage{subcaption}
\usepackage{float}
\usepackage{enumitem}
\usepackage{siunitx}
\usepackage[authoryear,round]{natbib}
\usepackage{xcolor}
\usepackage{hyperref}
\hypersetup{
  colorlinks=true,
  linkcolor=blue!55!black,
  citecolor=blue!55!black,
  urlcolor=blue!55!black,
  pdftitle={Improving Swaption Calibration in Factor HJM Stochastic Volatility Models},
  pdfauthor={Bram Brongers},
  pdfsubject={A First-Order Correction to Frozen Swap-Rate Loadings},
  pdfkeywords={factor HJM, stochastic volatility, swaption calibration, Cheyette model, frozen coefficients, first variation},
  pdfauthor={Bram Brongers}
}
\usepackage[capitalise,noabbrev]{cleveref}
\crefname{proposition}{Proposition}{Propositions}
\Crefname{proposition}{Proposition}{Propositions}
\usepackage{fancyhdr}
\usepackage{lastpage}
\graphicspath{{figures/}}
\allowdisplaybreaks
\numberwithin{equation}{section}
\setlist[itemize]{topsep=0.25em,itemsep=0.15em,parsep=0pt}
\setlist[enumerate]{topsep=0.25em,itemsep=0.15em,parsep=0pt}

\newtheorem{theorem}{Theorem}[section]
\newtheorem{proposition}{Proposition}[section]

\theoremstyle{definition}

\theoremstyle{remark}
\newtheorem{remark}[theorem]{Remark}

\newcommand{\E}{\mathbb{E}}
\newcommand{\Q}{\mathbb{Q}}
\newcommand{\R}{\mathbb{R}}
\newcommand{\dd}{\mathrm{d}}
\newcommand{\ii}{\mathrm{i}}
\newcommand{\cA}{\mathcal{A}}
\newcommand{\cD}{\mathcal{D}}
\newcommand{\cL}{\mathcal{L}}
\newcommand{\cV}{\mathcal{V}}

\newcommand{\exact}{\mathrm{ex}}
\newcommand{\frozen}{\mathrm{fr}}
\newcommand{\linear}{\mathrm{lin}}
\newcommand{\QuoteResidualThresholdBp}{0.05}
\newcommand{\FullModelRepricingMaximumSeBp}{0.0070}

\newcommand{\ModerateFrozenHeldoutBp}{0.3169}
\newcommand{\ModerateCorrectedHeldoutBp}{0.0118}
\newcommand{\ModerateHoldoutImprovementFactor}{26.8}
\newcommand{\NegativeFrozenHeldoutBp}{0.2622}
\newcommand{\NegativeCorrectedHeldoutBp}{0.0323}
\newcommand{\NegativeHoldoutImprovementFactor}{8.1}
\newcommand{\TwoBucketFrozenHeldoutBp}{0.1102}
\newcommand{\TwoBucketCorrectedHeldoutBp}{0.0658}
\newcommand{\TwoBucketHoldoutImprovementFactor}{1.7}
\newcommand{\ModerateMaximumQuoteSeBp}{0.0061}
\newcommand{\NegativeMaximumQuoteSeBp}{0.0063}
\newcommand{\TwoBucketMaximumQuoteSeBp}{0.0040}
\newcommand{\ModerateWeakestLiftPercent}{20.4}

\newcommand{\QmcReplicateCount}{8}
\newcommand{\MarketReplicateCount}{20}
\newcommand{\MarketQuoteNoiseBp}{0.25}
\newcommand{\StateStressResidualBp}{0.0981}
\newcommand{\BroadStressResidualBp}{0.0738}
\newcommand{\StateStressPdeFitBp}{0.1155}
\newcommand{\StateStressPdeHeldoutBp}{0.1043}
\newcommand{\BroadStressPdeFitBp}{0.0882}
\newcommand{\BroadStressPdeHeldoutBp}{0.0809}
\newcommand{\BroadEoneFrozenHeldoutBp}{1.7466}
\newcommand{\BroadEoneCorrectedHeldoutBp}{7.5129}
\newcommand{\BroadEoneRefinementBp}{2844.85}
\newcommand{\BroadSpectralFinalQuadratureBp}{0.000351}
\newcommand{\GammaSlopeMinimum}{1.974}
\newcommand{\GammaSlopeMaximum}{2.071}
\newcommand{\BoundaryLowerU}{0.55}
\newcommand{\BoundaryUpperU}{0.60}
\newcommand{\BoundaryLowerResidualBp}{0.0486}
\newcommand{\BoundaryUpperResidualBp}{0.0529}
\newcommand{\BoundaryLowerResidualSeBp}{0.0051}
\newcommand{\BoundaryUpperResidualSeBp}{0.0055}
\newcommand{\BoundaryEoneLastU}{0.2}
\newcommand{\BoundaryPdeFirstU}{0.4}

\newcommand{\NonflatFrozenHeldoutBp}{0.3174}
\newcommand{\NonflatCorrectedHeldoutBp}{0.0104}
\newcommand{\NonflatHoldoutImprovementFactor}{30.7}
\newcommand{\NonflatTrueRho}{-0.268}
\newcommand{\NonflatFrozenRho}{-0.123}
\newcommand{\NonflatCorrectedRho}{-0.269}
\newcommand{\NonflatMaximumResidualBp}{0.0141}

\newcommand{\NonflatMaximumQuoteSeBp}{0.0067}

\newcommand{\ThreeBySevenAtmEquivalentBp}{0.0028}

\newcommand{\TenByTenAtmEquivalentBp}{0.0103}

\title{
\textbf{Improving Swaption Calibration in Factor HJM Stochastic Volatility Models}\\[0.35em]
\large A First-Order Correction to Frozen Swap-Rate Loadings
}
\author{
Bram Brongers\\
\href{mailto:brambrongers98@gmail.com}{\texttt{brambrongers98@gmail.com}}
}
\date{27 August 2026}

\begin{document}
\maketitle

\begin{abstract}
The factor HJM stochastic volatility model of \citet{SeppRakhmonovFHJM2025}
obtains tractable swaption pricing by freezing the nonlinear swap-rate loading
along a deterministic expected-state path. This removes the dependence of
conditional swap-rate variance on the current yield-curve state. We introduce
a first-order Taylor correction to the loading that adds no calibration
parameters. Conditional on retaining the frozen annuity-measure drift, we show that the
first variation of the swap-rate transform is affine in the centered rate
states and reduces to one-dimensional equations in volatility. For
quadratic-drift lognormal stochastic volatility, these equations yield a
finite-dimensional ODE representation and a direct log-volatility formulation. Calibrations to independently generated nonlinear-model prices show
substantially lower stochastic volatility parameter bias and held-out pricing
error, with only modest changes in local calibration identifiability.

\end{abstract}

\noindent\textbf{Keywords:} factor HJM; Cheyette model; stochastic volatility; swaption calibration; frozen coefficients; first variation; Fourier transform.\\
\textbf{JEL classification:} G13, C63.

\section{Introduction}\label{sec:introduction}

The Heath--Jarrow--Morton framework models the evolution of the full forward-rate curve subject to a no-arbitrage drift restriction \citep{HeathJarrowMorton1992}. Finite-dimensional realizations are attractive in applications because they combine the consistency of HJM with a small Markov state. The Cheyette representation is a standard example \citep{Cheyette1992,AndersenPiterbarg2010}. More recently, the factor HJM construction of \citet{LyashenkoGoncharov2023} starts from an interpretable factor representation under the statistical measure and augments it by auxiliary factors so that the risk-neutral curve dynamics remain arbitrage-free and fit the initial term structure by construction. Exponential-polynomial bases, including Nelson--Siegel-type specifications, lead to tractable finite-dimensional models \citep{NelsonSiegel1987,ChristensenDieboldRudebusch2009}.

Stochastic volatility is relevant when a term-structure model is required to reproduce the skew and curvature of interest-rate option prices. \citet{SeppRakhmonovFHJM2025} augment the factor HJM model with a scalar stochastic volatility driver correlated with the curve factors. They derive annuity-measure swap-rate dynamics, Fourier valuation formulas, an exact transform for a CIR volatility specification, and a semi-analytic approximation for quadratic-drift lognormal volatility. Their empirical results show that the resulting Nelson--Siegel factor model can fit swaption and SOFR-option implied volatilities across expiries and tenors. A public reference implementation accompanies their paper \citep{StochVolModelsRepository}.

Bond prices, annuities, and swap rates are explicit functions of the finite-dimensional state. The obstacle is that the diffusion loading of the swap rate under its annuity measure is nonlinear in those states. In the notation of \citet{SeppRakhmonovFHJM2025},
\begin{equation}
  \dd S_t
  = \sigma_t\,\Lambda(t,X_t,Y_t)^{\top}\dd W_t^A,
  \qquad
  \Lambda(t,X,Y)=C(t)^{\top}\nabla_X S(t,X,Y).
  \label{eq:intro_exact_loading}
\end{equation}
The annuity-measure drift contains a second nonlinear coefficient,
\begin{equation}
  L_X(t,X,Y)=\nabla_X\log A(t,X,Y).
  \label{eq:intro_annuity_loading}
\end{equation}
To obtain a low-dimensional transform, the original construction evaluates both
coefficients along deterministic moment-closure paths
$(\bar X_t,\bar Y_t)$ approximating the annuity-measure expected states,
\begin{equation}
  (X_t,Y_t)\longmapsto(\bar X_t,\bar Y_t).
  \label{eq:intro_freeze}
\end{equation}
In particular, the swap-rate loading becomes the deterministic function
\begin{equation}
  \Lambda_0(t)=C(t)^{\top}\nabla_X S(t,\bar X_t,\bar Y_t).
  \label{eq:intro_constant_loading}
\end{equation}
This is closely related to the coefficient-freezing approximations commonly used to derive tractable market-model pricing formulas \citep{BrigoMercurio2007}.

The approximation in \cref{eq:intro_constant_loading} is analytically convenient, but it removes the dependence of the swap-rate diffusion loading, and hence of the conditional swap-rate variance, on the current yield-curve state. Calibration may then compensate by shifting the stochastic volatility parameters. These parameters can reproduce some of the resulting effects on the swaption smile even though the underlying mechanism is different. A good fit to option prices therefore does not guarantee recovery of the dynamics that generated them. Synthetic experiments make the difference observable because the data-generating parameters are known.

We replace \cref{eq:intro_constant_loading} by the first-order Taylor approximation
\begin{equation}
  \Lambda_{\linear}(t,Z)
  =\Lambda_0(t)+J(t)(Z-\bar Z_t),
  \qquad Z=(X,Y),
  \label{eq:intro_linear_loading}
\end{equation}
where \(J(t)\) is the Jacobian of \(C^{\top}\nabla_XS\) along the same deterministic centering path. This adds no free parameters, since \(J(t)\) is determined by the initial curve, the swap cash-flow dates, the factor basis, and the existing model parameters.

A direct substitution of \cref{eq:intro_linear_loading} into the swap-rate variance produces quadratic state terms. The complete state-linear model is therefore not affine in the usual sense of \citet{DuffieFilipovicSchachermayer2003}. Nevertheless, the first-order correction to the frozen transform remains tractable and can be computed without solving the full state-dimensional PDE. Under the same frozen annuity-measure drift used by the original approximation, the first variation of the transform is affine in the centered rate states, with coefficient functions governed by one-dimensional volatility equations. For quadratic-drift lognormal volatility, these equations can be coupled to the original first-order affine expansion, denoted E1 below, and projected onto a small polynomial basis.

\subsection{Related work}
Approximate swaption valuation in term-structure models has long relied on
low-dimensional representations of coupon-bond or swap-rate payoffs.
\citet{SingletonUmantsev2002} recover Fourier-based prices for options on
coupon-bearing instruments in affine term-structure models by approximating
the exercise probabilities. More closely related to the present setting,
\citet{SchragerPelsser2006} derive approximate swap-rate dynamics whose
volatility is affine in the term-structure factors, so that affine transform
methods remain available. The approximation presented in this paper has a similar objective,
but starts from the nonlinear swap-rate loading in the factor HJM
stochastic volatility model. The complete state-linear loading does not
preserve affine diffusion structure; instead, tractability is recovered at the level of its first variation.

Coefficient freezing followed by a correction also appears in market-model
approximations. \citet{Henrard2010} combines an initial freeze with a
corrector approximation to obtain explicit, strike-dependent swaption prices
in a local-volatility LMM. Perturbative methods have likewise been used to
introduce smile dynamics while retaining fast pricing.
\citet{AhdidaAlfonsiPalidda2017} expand caplet and swaption prices around
the linear Gaussian model in an affine stochastic-covariance extension,
while \citet{GairatGorovoyShcherbakov2026} derive perturbative
local-volatility approximations for Cheyette-type models.
In contrast to these approaches, we study specifically the error induced by
the expected-state freeze in the FHJM-SV swap-rate loading and derive the resulting first-variation transform.

\subsection{Contributions}
The contributions of this paper are as follows.

\begin{enumerate}[label=(\roman*)]

\item We study the calibration effect of the frozen-loading approximation using synthetic swaption prices generated independently from the full nonlinear model. By comparing frozen and state-linear specifications, we quantify the effects of loading linearization, freezing of the annuity-measure drift, first-order truncation, and the E1 approximation on calibration accuracy and parameter recovery.

\item Conditional on the prescribed centering path and frozen annuity drift, we derive the first variation of the swap-rate transform within the auxiliary state-linear loading family. In the scalar Cheyette model, the resulting correction is characterized by a system of one-dimensional parabolic equations in the volatility state. A matrix formulation shows that the same reduction extends to multiple curve factors when the frozen rate-state dynamics have affine drift.

\item For quadratic-drift lognormal stochastic volatility, we obtain a finite-dimensional implementation by combining the first-variation equations with the E1 approximation of \citet{SeppRakhmonovLogSV2023,SeppRakhmonovFHJM2025}. We also develop a Chebyshev collocation scheme in log volatility that does not rely on the finite polynomial representation. The derivation and numerical implementation are checked independently using symbolic identities, an analytically tractable Gaussian limit, finite-difference solutions of the one-dimensional PDEs, and pathwise randomized-QMC estimators.

\item We assess the first-order correction against the nonlinear model. In the parameter regimes used for the main calibration experiments, the correction substantially reduces approximation-induced parameter bias and held-out pricing errors. The stress tests identify cases in which either the finite E1 representation or the first-order approximation itself no longer meets the accuracy criteria.

\end{enumerate}

The numerical results indicate that the state-linear correction primarily reduces approximation-induced calibration bias, but does not materially change conditioning of the calibration problem. In the Moderate regime, the smallest singular value of the calibration Jacobian increases modestly, while the associated weak parameter direction is essentially unchanged; under Negative skew, the singular-value spectrum is largely unaffected. Within the parameter regimes for which the first-order approximation is accurate, the correction nevertheless produces substantially better parameter recovery and lower held-out pricing errors without introducing additional calibration parameters.

The remainder of the paper is organized as follows. \Cref{sec:model} reviews the factor HJM swaption setting and states the scalar lognormal-SV specialization used for the proof. \Cref{sec:first_variation} derives the first-variation equations and the finite ODE approximation. \Cref{sec:numerics_design} describes the calibration and transform-validation experiments. \Cref{sec:numerical_results} reports the results, including the stress tests. \Cref{sec:discussion} discusses interpretation and implementation limits, and \cref{sec:conclusion} concludes. Additional algebra and numerical details are collected in the appendices.

\section{Factor HJM swaptions and the frozen-coefficient approximation}\label{sec:model}

\subsection{Finite-dimensional factor HJM dynamics}

We begin with the notation needed for the swaption problem. Let \(X_t\in\R^d\) denote the principal yield-curve factors and \(Y_t\in\R^m\) the auxiliary factors required by the risk-neutral HJM drift. With an exponential-polynomial basis, the instantaneous forward curve can be written as
\begin{equation}
  f_t(\tau)=B(\tau)X_t+\widetilde B(\tau)Y_t+\widehat f_t(\tau),
  \label{eq:fwd_factor_rep}
\end{equation}
where \(\tau\) is time to maturity and \(\widehat f_t\) reproduces the observed initial curve. Under the risk-neutral measure \(\Q\), a scalar stochastic volatility extension has dynamics of the form
\begin{align}
  \dd X_t &= D X_t\,\dd t+\sigma_t C(t)\,\dd W_t^{(0)},
  \label{eq:general_X_Q}\\
  \dd Y_t &= \widetilde D Y_t\,\dd t+\sigma_t^2\widetilde\Omega(t)\,\dd t,
  \label{eq:general_Y_Q}\\
  \dd\sigma_t
  &=b_\sigma(t,\sigma_t)\,\dd t
    +\varsigma(t,\sigma_t)
      \left(\bm\beta(t)^{\top}\dd W_t^{(0)}+\beta_0(t)\dd B_t\right),
  \label{eq:general_sigma_Q}
\end{align}
where \(B\) is independent of \(W^{(0)}\), and \(\widetilde\Omega\) is fixed by the HJM no-arbitrage restriction. See \citet{LyashenkoGoncharov2023,SeppRakhmonovFHJM2025} for the construction and admissible bases.

The zero-coupon bond price is exponential-affine in \((X,Y)\):
\begin{equation}
  P(t,t+\tau)
  =\frac{P(0,t+\tau)}{P(0,t)}
   \exp\!\left[-B_P(\tau)X_t-\widetilde B_P(\tau)Y_t\right],
  \label{eq:bond_price}
\end{equation}
where
\begin{equation}
  B_P(\tau)=\int_0^\tau B(u)\,\dd u,
  \qquad
  \widetilde B_P(\tau)=\int_0^\tau\widetilde B(u)\,\dd u.
\end{equation}
For payment dates \(T_0<T_1<\cdots<T_N\) and accrual fractions \(\delta_n=T_{n+1}-T_n\), define
\begin{align}
  A(t,X,Y)&=\sum_{n=0}^{N-1}\delta_n P(t,T_{n+1};X,Y),
  \label{eq:annuity}\\
  S(t,X,Y)&=\frac{P(t,T_0;X,Y)-P(t,T_N;X,Y)}{A(t,X,Y)}.
  \label{eq:swap_rate}
\end{align}
All derivatives of \(A\) and \(S\) required below are therefore available analytically or by algorithmic differentiation.

\subsection{Annuity-measure dynamics and nonlinear coefficients}

Let \(\Q^A\) be the measure associated with the swap annuity. Define
\begin{equation}
  L_X(t,X,Y)=\nabla_X\log A(t,X,Y).
  \label{eq:Lx_definition}
\end{equation}
A Girsanov change of measure gives the state dynamics under \(\Q^A\):
\begin{align}
  \dd X_t
  &=\left[D X_t+\sigma_t^2 C C^{\top}L_X(t,X_t,Y_t)\right]\dd t
    +\sigma_t C\,\dd W_t^A,
  \label{eq:X_annuity_general}\\
  \dd Y_t
  &=\left[\widetilde D Y_t+\sigma_t^2\widetilde\Omega\right]\dd t,
  \label{eq:Y_annuity_general}
\end{align}
with the corresponding drift adjustment in the volatility process. The swap rate is a local martingale under \(\Q^A\), with diffusion
\begin{equation}
  \dd S_t
  =\sigma_t\,\Lambda(t,X_t,Y_t)^{\top}\dd W_t^A,
  \qquad
  \Lambda(t,X,Y)=C(t)^{\top}\nabla_XS(t,X,Y).
  \label{eq:swap_annuity_exact}
\end{equation}
Both \(L_X\) and \(\Lambda\) are nonlinear functions of the state because they contain ratios and weighted sums of exponential-affine bond prices.

For analytical tractability, \citet{SeppRakhmonovFHJM2025} replace these coefficients by their values along deterministic expected-state paths. In particular,
\begin{align}
  L_0(t)&=L_X(t,\bar X_t,\bar Y_t),
  \label{eq:L0}\\
  \Lambda_0(t)&=C(t)^{\top}\nabla_XS(t,\bar X_t,\bar Y_t),
  \label{eq:Lambda0}
\end{align}
where \((\bar X,\bar Y)\) are obtained from a deterministic moment closure. The projected swap rate then satisfies
\begin{equation}
  \dd s_t=\sigma_t\Lambda_0(t)^{\top}\dd W_t^A.
  \label{eq:frozen_projected_swap}
\end{equation}
This reduction concentrates all stochastic non-Gaussianity in the volatility driver and leads to a low-dimensional transform. The notation ``frozen loading'' below refers specifically to \cref{eq:Lambda0}; the annuity-measure drift is also frozen unless stated otherwise.

\subsection{Quadratic-drift lognormal volatility}

The lognormal stochastic volatility specification considered in the original model is
\begin{equation}
  \dd\sigma_t
  =(\kappa_1+\kappa_2\sigma_t)(\theta-\sigma_t)\,\dd t
   +\sigma_t\left(\bm\beta(t)^{\top}\dd W_t^{(0)}+\beta_0(t)\dd B_t\right).
  \label{eq:logsv_Q}
\end{equation}
The quadratic drift is useful for positivity, mean reversion, and stability under changes of numeraire; see \citet{KarasinskiSepp2012,CarrWillems2020,SeppRakhmonovLogSV2023}. Under \(\Q^A\), freezing \(L_X\) preserves the quadratic-lognormal form, with a deterministic adjustment of the quadratic drift.

The proof and numerical validation will use the scalar Cheyette specialization. Let \(\alpha(t)\) be the deterministic scalar factor loading and \(\lambda>0\) the rate mean reversion. Under the frozen annuity coefficient \(L_0(t)\),
\begin{align}
  \dd X_t
  &=\left[Y_t-\lambda X_t+\alpha_t^2\sigma_t^2L_0(t)\right]\dd t
    +\alpha_t\sigma_t\,\dd W_t,
  \label{eq:scalar_X}\\
  \dd Y_t
  &=\left[\alpha_t^2\sigma_t^2-2\lambda Y_t\right]\dd t,
  \label{eq:scalar_Y}\\
  \dd\sigma_t
  &=\left[(\kappa_1+\kappa_2\sigma_t)(\theta-\sigma_t)
          +\alpha_t\beta_t\sigma_t^2L_0(t)\right]\dd t
    +\sigma_t\left(\beta_t\dd W_t+\epsilon_t\dd B_t\right),
  \label{eq:scalar_sigma}
\end{align}
where \(W\) and \(B\) are independent. We set
\begin{equation}
  \nu_t^2=\beta_t^2+\epsilon_t^2.
  \label{eq:nu_definition}
\end{equation}
The scalar swap-rate loading is
\begin{equation}
  a(t,x,y)=\alpha(t)S_x(t,x,y).
  \label{eq:scalar_loading}
\end{equation}
Along the expected-state path, define
\begin{equation}
  a_0(t)=\alpha_t S_x(t,\bar x_t,\bar y_t),\qquad
  a_x(t)=\alpha_t S_{xx}(t,\bar x_t,\bar y_t),\qquad
  a_y(t)=\alpha_t S_{xy}(t,\bar x_t,\bar y_t).
  \label{eq:loading_derivatives}
\end{equation}
The proposed state-linear loading is
\begin{equation}
  a_{\linear}(t,x,y)
  =a_0(t)+a_x(t)(x-\bar x_t)+a_y(t)(y-\bar y_t).
  \label{eq:state_linear_scalar}
\end{equation}
No new parameter appears in \cref{eq:state_linear_scalar}.

\subsection{Why the complete state-linear model is not affine}

Let \(\xi=x-\bar x_t\) and \(\eta=y-\bar y_t\). The instantaneous variance of the projected swap rate under \cref{eq:state_linear_scalar} is
\begin{equation}
  \sigma_t^2\left[a_0+a_x\xi_t+a_y\eta_t\right]^2.
  \label{eq:quadratic_swap_variance}
\end{equation}
The square contains terms \(\sigma^2\xi^2\), \(\sigma^2\xi\eta\), and \(\sigma^2\eta^2\). Cross-covariances between the swap rate, \(X\), and \(\sigma\) also become state dependent. Consequently, the full state-linear model does not possess the exponential-affine transform of an affine diffusion. This does not preclude a controlled first-order correction around the frozen loading, which is derived next.

\section{First variation of the swap-rate transform}\label{sec:first_variation}

\subsection{Centered dynamics and the auxiliary scaling parameter}

Define centered states and shifted volatility by
\begin{equation}
  \xi_t=X_t-\bar x_t,\qquad
  \eta_t=Y_t-\bar y_t,\qquad
  v_t=\sigma_t-\theta,
  \qquad \sigma_t=\theta+v_t.
  \label{eq:centered_states}
\end{equation}
Introduce an auxiliary scalar \(\gamma\) multiplying the state-dependent part of the loading,
\begin{equation}
  a_\gamma(t,\xi,\eta)
  =a_0(t)+\gamma h(t,\xi,\eta),
  \qquad
  h(t,\xi,\eta)=a_x(t)\xi+a_y(t)\eta.
  \label{eq:gamma_loading}
\end{equation}
The frozen model is \(\gamma=0\), and \(\gamma=1\) gives the complete state-linear loading. The parameter \(\gamma\) organizes the perturbation expansion and is not calibrated.

Write
\begin{equation}
  m(t)=\alpha_t\beta_tL_0(t).
\end{equation}
The volatility drift in \(v\) is
\begin{equation}
  b_v(t,v)=k_0(t)-k_1(t)v-k_2(t)v^2,
  \label{eq:bv}
\end{equation}
with
\begin{equation}
  k_0=m\theta^2,
  \qquad
  k_1=\kappa_1+\kappa_2\theta-2m\theta,
  \qquad
  k_2=\kappa_2-m.
  \label{eq:k_coeffs}
\end{equation}
Let \(m_\sigma(t)\) denote the deterministic mean-volatility closure used to
construct the centering path. The implementation evolves this proxy for
\(\E[\sigma_t]\) and uses its square,
\begin{equation}
  m_\sigma(t)^2 \approx \bigl(\E[\sigma_t]\bigr)^2,
  \label{eq:mean_volatility_closure}
\end{equation}
in the expected-state equations; it is not an approximation to
\(\E[\sigma_t^2]\). The centered rate-state dynamics are
\begin{align}
  \dd\xi_t
  &=\left[\eta_t-\lambda\xi_t+r_x(t,v_t)\right]\dd t
    +\alpha_t\sigma_t\dd W_t,
  \label{eq:centered_xi}\\
  \dd\eta_t
  &=\left[-2\lambda\eta_t+r_y(t,v_t)\right]\dd t,
  \label{eq:centered_eta}
\end{align}
where
\begin{equation}
  r_y(t,v)=\alpha_t^2\left[(\theta+v)^2-m_\sigma(t)^2\right],
  \qquad
  r_x(t,v)=L_0(t)r_y(t,v).
  \label{eq:rxy}
\end{equation}
The derivation remains valid for any deterministic centering path if \(r_x,r_y\) are replaced by the corresponding deterministic residuals.

The projected swap rate satisfies
\begin{equation}
  \dd S_t=\sigma_t a_\gamma(t,\xi_t,\eta_t)\dd W_t.
  \label{eq:S_gamma}
\end{equation}
Because the annuity-measure coefficient is still frozen, the state process \((\xi,\eta,v)\) is independent of \(\gamma\). This triangular structure is the key to the first-order reduction.

\subsection{Fourier-transformed backward equation}

Let \(\phi\in\mathbb C\) lie in a strip on which the relevant exponential
moment exists.  We use the bilateral-Laplace convention
\begin{equation}
  f_\gamma(\tau,\xi,\eta,v;\phi)
  =\E^A\!\left[
       e^{-\phi(S_T-S_t)}
       \mid \xi_t=\xi,\eta_t=\eta,v_t=v
     \right],
  \qquad \tau=T-t,
  \label{eq:transform_definition}
\end{equation}
where \(t=T-\tau\).  On the imaginary axis this is the characteristic
function under \(\phi=-\ii u\); displaced contours provide the exponential
damping used in Fourier option pricing.  At expiry no increment remains, so
\begin{equation}
  f_\gamma(0,\xi,\eta,v;\phi)=1.
  \label{eq:transform_terminal}
\end{equation}

It is useful to state explicitly how the transformed operator arises.  Let
\[
  U_\gamma(t,s,\xi,\eta,v;\phi)
  =\E^A\!\left[
       e^{-\phi S_T}
       \mid S_t=s,\xi_t=\xi,\eta_t=\eta,v_t=v
     \right].
\]
Because the coefficients of the joint dynamics do not depend on the level of
the swap rate,
\[
  U_\gamma(t,s,\xi,\eta,v;\phi)
  =e^{-\phi s}F_\gamma(t,\xi,\eta,v;\phi),
  \qquad
  F_\gamma(t,\cdot;\phi)=f_\gamma(T-t,\cdot;\phi).
\]
Substituting this factorization into the backward Kolmogorov equation for the
joint process \((S,\xi,\eta,v)\), and dividing by \(e^{-\phi s}\), gives
\begin{equation}
  \partial_tF_\gamma
  +\cL_{\gamma,\phi,t}F_\gamma=0,
  \qquad
  F_\gamma(T,\xi,\eta,v;\phi)=1.
  \label{eq:calendar_transform_equation}
\end{equation}
Equivalently, in time to expiry,
\begin{equation}
  \partial_\tau f_\gamma
  =\cL_{\gamma,\phi,T-\tau}f_\gamma,
  \qquad f_\gamma(0,\xi,\eta,v;\phi)=1.
  \label{eq:tau_transform_equation}
\end{equation}
Thus every deterministic input to the backward equation is read in reversed
calendar time:
\[
  \overleftarrow c(\tau):=c(T-\tau).
\]
We suppress the arrow notation in what follows, with all deterministic coefficients understood to be evaluated at \(T-\tau\). Appendix~\ref{app:phase1_time} gives the corresponding time-change derivation and a numerical consistency check for nonconstant coefficient paths.

Writing \(\sigma=\theta+v\), \(h=a_x\xi+a_y\eta\), and
\(\nu^2=\beta^2+\epsilon^2\), the transformed backward operator is
\begin{equation}
\boxed{
\begin{aligned}
  \cL_{\gamma,\phi}f
  ={}&b_\xi f_\xi+b_\eta f_\eta+b_v f_v\\
  &+\frac12\sigma^2\left(
       \alpha^2f_{\xi\xi}+\nu^2f_{vv}+2\alpha\beta f_{\xi v}
     \right)\\
  &-\phi\sigma^2(a_0+\gamma h)
       \left(\alpha f_\xi+\beta f_v\right)
   +\frac12\phi^2\sigma^2(a_0+\gamma h)^2 f,
\end{aligned}}
  \label{eq:full_generator}
\end{equation}
where
\begin{equation}
  b_\xi=\eta-\lambda\xi+r_x,
  \qquad
  b_\eta=-2\lambda\eta+r_y.
\end{equation}
The first line contains the state drifts.  The second line is the diffusion
operator for \((\xi,v)\), including their covariance
\(\dd\langle\xi,v\rangle_t=\alpha\beta\sigma^2\dd t\).  Conjugation by
\(e^{-\phi s}\) converts the \(S\)--\(\xi\) and \(S\)--\(v\)
cross-derivatives into the terms proportional to \(-\phi\); their sign follows
from \(\partial_s e^{-\phi s}=-\phi e^{-\phi s}\).  The final term is
\(\tfrac12\phi^2\dd\langle S\rangle_t/\dd t\).

\pagebreak[2]
Since \(a_\gamma=a_0+\gamma h\), the operator has the exact polynomial
decomposition
\begin{equation}
  \cL_{\gamma,\phi}
  =\cL_{0,\phi}+\gamma\cL_{1,\phi}+\gamma^2\cL_{2,\phi},
  \label{eq:generator_expansion}
\end{equation}
where
\begin{align}
  \cL_{1,\phi}f
  &=-\phi\sigma^2h(\alpha f_\xi+\beta f_v)
    +\phi^2\sigma^2a_0h f,
  \label{eq:L1}\\
  \cL_{2,\phi}f
  &=\frac12\phi^2\sigma^2h^2 f.
  \label{eq:L2}
\end{align}
This decomposition is exact at the generator level.  The first variation at
\(\gamma=0\) depends on \(\cL_{1,\phi}\); the quadratic term
\(\cL_{2,\phi}\) enters only from second order onward.  The approximation
introduced below is the truncation of the solution in \(\gamma\), not the
operator identity above.

\subsection{Exact affine-state closure of the first variation}

Under the second-order moment and transform-domain conditions of Theorem~\ref{thm:transform_differentiability}, Taylor's theorem gives
\begin{equation}
  f_\gamma=f_0+\gamma f_1+O(\gamma^2).
  \label{eq:f_expansion}
\end{equation}
Matching powers of \(\gamma\) gives
\begin{align}
  \partial_\tau f_0&=\cL_{0,\phi}f_0,
  &f_0(0)&=1,
  \label{eq:f0_pde}\\
  \partial_\tau f_1&=\cL_{0,\phi}f_1+\cL_{1,\phi}f_0,
  &f_1(0)&=0.
  \label{eq:f1_pde}
\end{align}
At \(\gamma=0\), the terminal condition and all coefficients relevant to the swap transform are independent of \((\xi,\eta)\). Hence \(f_0=f_0(\tau,v;\phi)\). Define the one-dimensional frozen-loading operator
\begin{equation}
\boxed{
  \cA_\phi g
  =\frac12\nu^2\sigma^2g_{vv}
   +\left(b_v-\phi\beta a_0\sigma^2\right)g_v
   +\frac12\phi^2a_0^2\sigma^2g.}
  \label{eq:Aphi}
\end{equation}
Then \(\partial_\tau f_0=\cA_\phi f_0\).

\begin{proposition}[First variation of the swap-rate transform]\label{prop:first_variation}
Assume the strong-solution, moment, and transform-domain conditions in Appendix~\ref{app:phase1_differentiability}, 
and uniqueness in the relevant mild-solution growth class. Then the first variation
\(f_1=\left.\partial_\gamma f_\gamma\right|_{\gamma=0}\) exists as a mild
Feynman--Kac solution and is affine in the centered rate states.  If the
additional regularity or hypoellipticity conditions stated there hold, it is
also a classical solution:
\begin{equation}
  f_1(\tau,\xi,\eta,v)
  =g_0(\tau,v)+g_x(\tau,v)\xi+g_y(\tau,v)\eta.
  \label{eq:f1_affine}
\end{equation}
The coefficient functions solve
\begin{equation}
\boxed{
\begin{aligned}
  \partial_\tau g_x
  ={}&(\cA_\phi-\lambda)g_x
      +a_x\sigma^2\left(\phi^2a_0f_0-\phi\beta\partial_vf_0\right),\\[2pt]
  \partial_\tau g_y
  ={}&(\cA_\phi-2\lambda)g_y+g_x
      +a_y\sigma^2\left(\phi^2a_0f_0-\phi\beta\partial_vf_0\right),\\[2pt]
  \partial_\tau g_0
  ={}&\cA_\phi g_0
      +(r_x-\phi\alpha a_0\sigma^2)g_x+r_yg_y
      +\alpha\beta\sigma^2\partial_vg_x,
\end{aligned}}
  \label{eq:g_system}
\end{equation}
with \(g_0(0,v)=g_x(0,v)=g_y(0,v)=0\).
\end{proposition}

\begin{proof}
Under the frozen annuity drift the state process does not depend on \(\gamma\),
and, conditionally from time \(t\) to \(T\),
\begin{equation}
  \Delta S^\gamma_{t,T}
  =\Delta S^0_{t,T}+\gamma\mathcal I_{t,T},
  \qquad
  \mathcal I_{t,T}
  =\int_t^T\sigma_s\bigl[a_x(s)\xi_s+a_y(s)\eta_s\bigr]\dd W_s.
  \label{eq:proof_pathwise_increment}
\end{equation}
Dominated convergence under the moment assumptions of
\cref{thm:transform_differentiability} gives, on the characteristic-function
line,
\begin{equation}
  f_1(t,\xi,\eta,v)
  =-\phi\E^A\!\left[
      \mathcal I_{t,T}e^{-\phi\Delta S^0_{t,T}}
      \mid \xi_t=\xi,\eta_t=\eta,v_t=v
    \right].
  \label{eq:proof_pathwise_derivative}
\end{equation}
For complex \(\phi\), the same identity requires the exponential-moment
condition in \cref{eq:complex_transform_domain}; displaced Fourier contours
are used only where that transform-domain condition is assumed or verified.

Put \(Z=(\xi,\eta)^\top\). For fixed initial volatility and fixed future
volatility and Brownian realization, variation of constants in the frozen
affine rate-state equation gives
\begin{equation}
  Z_s=\Phi(s,t)Z_t+\zeta_{s,t},
  \label{eq:proof_affine_flow}
\end{equation}
where \(\Phi\) is the deterministic fundamental matrix and
\(\zeta_{s,t}\) is independent of the initial rate state. Consequently,
\begin{equation}
  \mathcal I_{t,T}=\mathcal I^0_{t,T}+q_{t,T}^{\top}Z_t.
\end{equation}
The frozen increment
\(\Delta S^0_{t,T}=\int_t^T\sigma_sa_0(s)\dd W_s\) is independent of the
initial centered rate state because the frozen volatility process is
autonomous with respect to \(Z_t\). Thus \cref{eq:proof_pathwise_derivative}
is affine in \(Z_t\), proving \cref{eq:f1_affine}.

Let \(U_{0,\phi}(\tau,s)\) be the time-inhomogeneous frozen-transform
evolution family on the relevant growth class. Duhamel's formula gives the
unique mild solution
\begin{equation}
  f_1(\tau)
  =\int_0^\tau U_{0,\phi}(\tau,s)
       \bigl[\cL_{1,\phi,s}f_0(s)\bigr]\dd s.
  \label{eq:f1_duhamel}
\end{equation}
The generator calculation provides an algebraic check of this pathwise
closure: for
\(\cV=\{g_0(v)+g_x(v)\xi+g_y(v)\eta\}\), one has
\(\cL_{0,\phi}\cV\subseteq\cV\), while
\begin{equation}
  \cL_{1,\phi}f_0
  =\sigma^2(a_x\xi+a_y\eta)
    \left(\phi^2a_0f_0-\phi\beta\partial_vf_0\right)\in\cV.
\end{equation}
Substitution of the already established affine representation into the mild
equation, or into \cref{eq:f1_pde} when classical derivatives exist, and
collection of the coefficients of \(1\), \(\xi\), and \(\eta\) gives
\cref{eq:g_system}. The \(+g_x\) term comes from the \(+\eta\) drift of
\(\xi\), and \(\alpha\beta\sigma^2\partial_vg_x\) comes from the mixed
\(\xi\)--\(v\) covariance. No quadratic state coefficient is generated.

The displayed PDEs hold classically only under the additional regularity or
hypoellipticity assumptions in Appendix~\ref{app:phase1_differentiability};
uniform ellipticity is not assumed.
\end{proof}

\begin{remark}
Even if the direct loading sensitivity \(a_y\) is neglected, \(g_y\) generally cannot be set to zero. The auxiliary state \(Y\) affects the drift of \(X\), which generates the \(+g_x\) coupling in the second equation of \cref{eq:g_system}.
\end{remark}

\begin{remark}
Conditional on the prescribed deterministic centering path and the frozen
annuity-measure drift, the affine dependence of the first variation on the
centered rate states is exact within the auxiliary state-linear loading
family. Relative to the original nonlinear factor-HJM model, the loading
Taylor remainder, evaluation at full loading scale \(\gamma=1\), and the E1
or numerical volatility representation remain separate approximation layers.
\end{remark}

\subsection{Log-transform equations}

Suppose \(f_0\neq0\) on the chosen contour and write
\begin{equation}
  f_0=e^A,
  \qquad
  \psi=\partial_vA.
  \label{eq:Apsi}
\end{equation}
Set
\begin{equation}
  f_1=f_0P,
  \qquad
  P=p_0+p_x\xi+p_y\eta.
  \label{eq:P_definition}
\end{equation}
Thus \(P=\left.\partial_\gamma\log f_\gamma\right|_{\gamma=0}\). Define
\begin{align}
  \chi_v
  &=b_v+\nu^2\sigma^2\psi-\phi\beta a_0\sigma^2,
  \label{eq:chi_v}\\
  \cD_\phi g
  &=\frac12\nu^2\sigma^2g_{vv}+\chi_vg_v,
  \label{eq:Dphi}\\
  \chi_x
  &=r_x+\alpha\beta\sigma^2\psi-\phi\alpha a_0\sigma^2,
  \label{eq:chi_x}\\
  q_x
  &=a_x\sigma^2(\phi^2a_0-\phi\beta\psi),
  \qquad
  q_y=a_y\sigma^2(\phi^2a_0-\phi\beta\psi).
  \label{eq:qxy}
\end{align}
Dividing \cref{eq:g_system} by \(f_0\) gives the equivalent system
\begin{equation}
\boxed{
\begin{aligned}
  \partial_\tau p_x
  &=\cD_\phi p_x-\lambda p_x+q_x,\\
  \partial_\tau p_y
  &=\cD_\phi p_y-2\lambda p_y+p_x+q_y,\\
  \partial_\tau p_0
  &=\cD_\phi p_0+\chi_xp_x+r_yp_y
    +\alpha\beta\sigma^2\partial_vp_x,
\end{aligned}}
  \label{eq:p_system}
\end{equation}
with zero initial conditions.

At inception, the deterministic centering path starts at the model state, so \(\xi_0=\eta_0=0\). If \(v_0=\sigma_0-\theta\), then
\begin{equation}
  \left.\partial_\gamma\log M_\gamma(\phi)\right|_{\gamma=0}
  =p_0(T,v_0;\phi).
  \label{eq:logmgf_derivative}
\end{equation}
The literal first-order reconstruction is
\begin{equation}
  M_\gamma(\phi)
  =M_0(\phi)\left[1+\gamma p_0(T,v_0;\phi)\right]+O(\gamma^2).
  \label{eq:mgf_reconstruction}
\end{equation}
Exponentiating the correction, \(M_0e^{\gamma p_0}\), agrees to first order but selects a particular set of higher-order terms. We use the additive form in the numerical validation.

\subsection{One-dimensional PDE in log volatility}

The exact first-variation equations are conveniently solved in log volatility
\begin{equation}
  \ell=\log\sigma,
  \qquad \sigma=e^\ell.
\end{equation}
Since
\begin{equation}
  \frac12\nu^2\sigma^2\partial_{vv}
  =\frac12\nu^2(\partial_{\ell\ell}-\partial_\ell),
\end{equation}
the volatility diffusion coefficient becomes constant. Define
\begin{equation}
  \mathfrak A_\phi u
  =\frac12\nu^2u_{\ell\ell}
   +\left(\frac{b_v}{\sigma}-\frac12\nu^2-\phi\beta a_0\sigma\right)u_\ell
   +\frac12\phi^2a_0^2\sigma^2u.
  \label{eq:logvol_operator}
\end{equation}
Then \cref{eq:g_system} becomes
\begin{align}
  \partial_\tau f_0
  &=\mathfrak A_\phi f_0,
  \label{eq:logvol_f0}\\
  \partial_\tau g_x
  &=(\mathfrak A_\phi-\lambda)g_x
    +a_x\left(-\phi\beta\sigma\partial_\ell f_0
               +\phi^2a_0\sigma^2f_0\right),
  \label{eq:logvol_gx}\\
  \partial_\tau g_y
  &=(\mathfrak A_\phi-2\lambda)g_y+g_x
    +a_y\left(-\phi\beta\sigma\partial_\ell f_0
               +\phi^2a_0\sigma^2f_0\right),
  \label{eq:logvol_gy}\\
  \partial_\tau g_0
  &=\mathfrak A_\phi g_0
    +(r_x-\phi\alpha a_0\sigma^2)g_x+r_yg_y
    +\alpha\beta\sigma\partial_\ell g_x.
  \label{eq:logvol_g0}
\end{align}
This system is used as the reference solution in the numerical section. It removes the finite polynomial approximation but still uses the frozen annuity drift, the linearized loading, and the first-order truncation in \(\gamma\).

\subsection{Finite ODE realization under the E1 volatility expansion}

The frozen quadratic-lognormal transform is not exactly affine. \citet{SeppRakhmonovLogSV2023,SeppRakhmonovFHJM2025} use a first-order affine expansion in the shifted volatility \(v=\sigma-\theta\),
\begin{equation}
  A(\tau,v;\phi)
  =A_0(\tau;\phi)+A_1(\tau;\phi)v+A_2(\tau;\phi)v^2,
  \label{eq:E1_ansatz}
\end{equation}
obtained by truncating the log-transform PDE at degree two. The three baseline ODEs are recorded in Appendix~\ref{app:E1}.

Insert
\begin{equation}
  \psi=A_1+2A_2v
\end{equation}
into \cref{eq:p_system} and approximate
\begin{equation}
  p_j(\tau,v)=\sum_{r=0}^N c_{j,r}(\tau)v^r,
  \qquad j\in\{0,x,y\}.
  \label{eq:polynomial_p}
\end{equation}
Let \(\Pi_N\) denote truncation to powers \(v^0,\ldots,v^N\). The projected equations are
\begin{align}
  \partial_\tau p_x
  &=\Pi_N\left[\cD_\phi^{E1}p_x-\lambda p_x+q_x^{E1}\right],
  \label{eq:proj_px}\\
  \partial_\tau p_y
  &=\Pi_N\left[\cD_\phi^{E1}p_y-2\lambda p_y+p_x+q_y^{E1}\right],
  \label{eq:proj_py}\\
  \partial_\tau p_0
  &=\Pi_N\left[\cD_\phi^{E1}p_0+\chi_x^{E1}p_x+r_yp_y
    +\alpha\beta\sigma^2\partial_vp_x\right].
  \label{eq:proj_p0}
\end{align}
The source \(q_x^{E1}\) has degree at most three, so \(N=3\) is the smallest degree that includes the source before projection.

For implementation, let \(D_N\) denote differentiation in the ascending monomial basis and let \(M_N[g]\) denote multiplication by a polynomial \(g\), followed by \(\Pi_N\). With
\begin{equation}
  K_N
  =\frac12\nu^2M_N[\sigma^2]D_N^2+M_N[\chi_v^{E1}]D_N,
  \label{eq:KN}
\end{equation}
the new coefficient vectors satisfy the block-linear system
\begin{equation}
\boxed{
\begin{aligned}
  \dot{\bm c}_x
  &=(K_N-\lambda I)\bm c_x+\bm q_x,\\
  \dot{\bm c}_y
  &=(K_N-2\lambda I)\bm c_y+\bm c_x+\bm q_y,\\
  \dot{\bm c}_0
  &=K_N\bm c_0+M_N[\chi_x^{E1}]\bm c_x+M_N[r_y]\bm c_y
    +\alpha\beta M_N[\sigma^2]D_N\bm c_x.
\end{aligned}}
  \label{eq:block_ode}
\end{equation}
For \(N=6\), the correction adds \(3(N+1)=21\) complex ODE coefficients to the three E1 coefficients. The baseline equations are unaffected by the correction, so the combined system is block triangular.

\begin{remark}[Exact and approximate parts]
Conditional on the prescribed centering path and frozen annuity drift, the affine dependence of the first variation on \((\xi,\eta)\) is exact within the auxiliary state-linear loading family. The full-scale first-order truncation and the finite polynomial dependence on \(v\) in \cref{eq:polynomial_p} are additional approximations. The polynomial representation can be highly accurate, but increasing \(N\) is not guaranteed to produce monotone convergence when the volatility distribution is broad.
\end{remark}

\subsection{Extension to multiple curve factors}

Let \(z\in\R^n\) be a centered vector containing the principal and auxiliary rate states, and suppose the frozen state dynamics have affine drift
\begin{equation}
  \dd z_t=(K(t)z_t+r(t,v_t))\dd t+\sigma_tG(t)\dd W_t.
\end{equation}
For a vector swap-rate loading
\begin{equation}
  a_\gamma(t,z)=a_0(t)+\gamma J(t)z,
\end{equation}
the source in the first-variation equation is linear in \(z\), and the frozen generator preserves functions of the form
\begin{equation}
  g_0(v)+g_z(v)^{\top}z.
\end{equation}
The scalar damping terms \(-\lambda g_x\), \(-2\lambda g_y\), and the coupling \(+g_x\) in \cref{eq:g_system} are replaced by the transpose of the rate-state drift matrix \(K(t)^{\top}\). The number of one-dimensional coefficient functions grows linearly with the centered-state dimension \(n\), including both principal and auxiliary rate states. The finite polynomial projection in volatility is unchanged in principle.

The full matrix-valued coefficient system and its proof are given in Appendix~\ref{app:phase1_multidimensional}.

\subsection{Approximation errors}

When comparing the corrected transform with the fully nonlinear factor HJM model, it is useful to distinguish five separate sources of approximation error:

\begin{enumerate}[label=(\arabic*)]

\item \textbf{Deterministic centering closure.} The state expansion is evaluated along a prescribed expected-state proxy, including the mean-volatility closure in \cref{eq:mean_volatility_closure}; that path is not asserted to equal the exact stochastic moments.

\item \textbf{Frozen annuity drift.} The state-dependent coefficient \(L_X(t,X,Y)\) is replaced by the deterministic coefficient \(L_0(t)\).

\item \textbf{Loading linearization.} The nonlinear swap-rate loading is replaced by its first-order Taylor approximation \(a_0+a_x\xi+a_y\eta\), thereby neglecting quadratic and higher-order dependence on the rate states.

\item \textbf{First-order perturbation.} The transform is truncated after first order in \(\gamma\). The terms \(\cL_{1,\phi}f_1\) and \(\cL_{2,\phi}f_0\), including the contribution proportional to \(h^2\), enter at order \(\gamma^2\).

\item \textbf{E1 and polynomial approximation.} The one-dimensional volatility equations for the first variation are approximated using the E1 representation and the finite polynomial expansion in \cref{eq:E1_ansatz,eq:polynomial_p}.

\end{enumerate}

The loading-error analysis in \cref{sec:loading_approximation_accuracy} quantifies the effect of the linearization in item~(3). The one-dimensional PDE provides a reference solution that does not use the E1 or finite polynomial approximations in item~(5), while the pathwise first-variation identities in Appendix~\ref{app:mc_identities} check the perturbation result independently. A separate comparison of frozen and state-dependent annuity drifts measures the contribution of item~(2); in the parameter regimes considered here, it is not the dominant source of the observed calibration bias.

\section{Numerical design}\label{sec:numerics_design}

We first measure the calibration bias caused by freezing the swap-rate loading. We then test whether the first-variation transform reproduces the pricing effect of restoring its state dependence, comparing it with a one-dimensional PDE, pathwise derivative estimates, and its finite E1 ODE representation. Finally, we calibrate the corrected Fourier pricer to swaption prices generated independently from the fully nonlinear model.

\subsection{Numerical implementation}\label{sec:reference_implementation}

The factor HJM conventions, expected-state coefficient freeze, quadratic-drift lognormal volatility model, E1 transform equations, and Fourier-pricing architecture follow \citet{SeppRakhmonovFHJM2025} and their reference implementation \citep{StochVolModelsRepository}.

Substituting \cref{eq:E1_ansatz} into the full frozen-loading log-transform PDE and truncating at degree two gives the baseline E1 right-hand side. At random complex coefficients it agrees with the reference E1 equations to approximately \(2\times10^{-14}\). Symbolic substitution in \cref{eq:g_system,eq:p_system} gives identically zero residuals for the coefficients of
\begin{equation}
  1,\quad \xi,\quad \eta,\quad \xi^2,\quad \xi\eta,\quad \eta^2.
\end{equation}
\subsection{Common scalar Cheyette specification}\label{sec:common_specification}

All experiments use the one-factor Cheyette bond representation
\begin{equation}
  P(t,T)
  =\frac{P(0,T)}{P(0,t)}
   \exp\!\left[-G(t,T)X_t-\frac12G(t,T)^2Y_t\right],
  \qquad
  G(t,T)=\frac{1-e^{-\lambda(T-t)}}{\lambda}.
  \label{eq:cheyette_bond_numerical}
\end{equation}
The continuously compounded initial zero rate is flat at \(2.5\%\), and fixed-leg payments are annual. The flat curve isolates the loading approximation and keeps the numerical experiments comparable; the derivation itself does not depend on this choice.

Unless stated otherwise, \(\theta=\sigma_0=1\). We report the instantaneous volatility magnitude and rate--volatility correlation as
\begin{equation}
  \nu=\sqrt{\beta^2+\epsilon^2},
  \qquad
  \rho=\frac{\beta}{\sqrt{\beta^2+\epsilon^2}}.
  \label{eq:nu_rho}
\end{equation}
The sign of \(\beta\) determines the direction of instantaneous dependence between the curve Brownian motion and the volatility process.

\subsection{Loading approximation and synthetic calibration}
\label{sec:calibration_design}

For each simulated state path, we evaluate the frozen, state-linear, and nonlinear swap-rate loadings,
\begin{align}
  a_{\frozen}(t)&=a_0(t),
  \label{eq:loading_frozen_num}\\
  a_{\linear}(t)&=a_0(t)+a_x(t)\xi_t+a_y(t)\eta_t,
  \label{eq:loading_linear_num}\\
  a_{\exact}(t)&=\alpha_t S_x(t,X_t,Y_t).
  \label{eq:loading_exact_num}
\end{align}
As a second-order diagnostic, define
\begin{align}
  a_{\mathrm{quad}}(t)
  ={}&a_0(t)+a_x(t)\xi_t+a_y(t)\eta_t
      +\frac12a_{xx}(t)\xi_t^2
      +a_{xy}(t)\xi_t\eta_t
      +\frac12a_{yy}(t)\eta_t^2.
  \label{eq:loading_quadratic_num}
\end{align}
The first loading derivatives use the closed-form bond and swap derivatives.
The Hessian entries are obtained by five-point centered differences of those
closed-form first derivatives and are checked for step-size convergence at
perturbed states. Let \(P_{\mathrm{quad}}\) denote the corresponding Monte
Carlo price.

The accuracy of the frozen and state-linear approximations is measured on common simulated paths using the variance-weighted relative errors
\begin{align}
  R_{\frozen}^2
  &=\frac{\E\!\left[\int_0^T\sigma_t^2
       |a_{\exact}(t)-a_{\frozen}(t)|^2\dd t\right]}
      {\E\!\left[\int_0^T\sigma_t^2|a_{\exact}(t)|^2\dd t\right]},
  \label{eq:Rfrozen}\\
  R_{\linear}^2
  &=\frac{\E\!\left[\int_0^T\sigma_t^2
       |a_{\exact}(t)-a_{\linear}(t)|^2\dd t\right]}
      {\E\!\left[\int_0^T\sigma_t^2|a_{\exact}(t)|^2\dd t\right]}.
  \label{eq:Rlinear}
  \\
  R_{\mathrm{quad}}^2
  &=\frac{\E\!\left[\int_0^T\sigma_t^2
       |a_{\exact}(t)-a_{\mathrm{quad}}(t)|^2\dd t\right]}
      {\E\!\left[\int_0^T\sigma_t^2|a_{\exact}(t)|^2\dd t\right]}.
  \label{eq:Rquadratic}
\end{align}

To make the contribution of the auxiliary state unambiguous, put
\(h_x=a_x\xi\), \(h_y=a_y\eta\), and \(h=h_x+h_y\), and define
\begin{equation}
  E_x=\frac{\E\int_0^T\sigma_t^2h_{x,t}^2\dd t}
             {\E\int_0^T\sigma_t^2h_t^2\dd t},\qquad
  E_y=\frac{\E\int_0^T\sigma_t^2h_{y,t}^2\dd t}
             {\E\int_0^T\sigma_t^2h_t^2\dd t},\qquad
  E_{xy}=\frac{2\E\int_0^T\sigma_t^2h_{x,t}h_{y,t}\dd t}
             {\E\int_0^T\sigma_t^2h_t^2\dd t}.
  \label{eq:loading_increment_energy}
\end{equation}
Thus \(E_x+E_y+E_{xy}=1\); in particular, \(E_y\) is not interpreted as an
additive fraction without the cross term. For each specification we report
the maxima of these coefficient diagnostics across the three calibrated
expiry--tenor pairs. The \(10\mathrm{Y}\times10\mathrm{Y}\) swaption is used
only for transform validation.

Synthetic option prices are generated with the pathwise nonlinear loading \cref{eq:loading_exact_num} and the nonlinear annuity-measure drift. Both calibration models use the expected-state frozen annuity drift, together with either \cref{eq:loading_frozen_num} or \cref{eq:loading_linear_num}. The goal of this comparison is to see whether restoring first-order state dependence in the swap loading removes a material part of the bias even when the measure-drift approximation is left unchanged.

After calibration, we also reprice each fitted parameter vector under the full nonlinear model. For \(m\in\{F,C\}\), denoting the frozen and corrected calibration estimates by \(\widehat\theta_m\), the diagnostic is
\begin{equation}
  \Delta P^{\mathrm{NL}}_{m,q}
  =P^{\mathrm{NL}}_q(\widehat\theta_m)
   -P^{\mathrm{NL}}_q(\theta^\star),
  \label{eq:full-model-repricing}
\end{equation}
where both prices use the nonlinear loading and nonlinear annuity-measure drift. The differences are estimated directly using common antithetic scrambled-Sobol paths for the fitted and data-generating parameter vectors, with first-order Richardson correction for time discretization. The path counts, randomizations, and time grids are chosen so that the resulting Monte Carlo error is small relative to the reported repricing differences.

The calibration experiments use three payer-swaption smiles,
\begin{equation}
  1\mathrm{Y}\times5\mathrm{Y},\qquad
  3\mathrm{Y}\times7\mathrm{Y},\qquad
  5\mathrm{Y}\times10\mathrm{Y}.
  \label{eq:calibration_swaptions}
\end{equation}
For each swaption \(i\), let \(F_i\) denote the initial forward swap rate and let
\[
  s_i=\widehat{\operatorname{sd}}(S_{T_i})
\]
denote the terminal swap-rate standard deviation estimated from a separate pilot simulation under the nonlinear data-generating model. Strikes are specified in standardized moneyness,
\begin{equation}
  K_i(m)=F_i+m\,s_i.
  \label{eq:calibration_moneyness}
\end{equation}

The holdout strikes are interlaced with the calibration grid and are not used in parameter estimation. They therefore provide a within-smile assessment of interpolation accuracy at strikes not included in the calibration objective.

Prices are converted to Bachelier implied volatilities for reporting. One normal-volatility basis point denotes \(10^{-4}\) in annualized normal volatility. For calibration quote \(q\), define
\begin{equation}
  \widetilde{\operatorname{Vega}}^N_q
  =
  \max\!\left(
    \operatorname{Vega}^N_q,\,
    0.2\,\operatorname{Vega}^N_{\mathrm{ATM},i(q)}
  \right),
  \label{eq:vega_weight}
\end{equation}
where \(i(q)\) denotes the corresponding swaption smile. The least-squares residual is
\begin{equation}
  r_q(\vartheta)
  =
  \frac{
    P_q(\vartheta)-P_q^\star
  }{
    \widetilde{\operatorname{Vega}}^N_q\,10^{-4}
  }.
  \label{eq:calibration_residual}
\end{equation}
This scaling expresses the calibration residual approximately in normal-volatility basis points while preventing very low-vega strikes from receiving excessive weight.

\Cref{tab:calibration-setup} summarizes the calibration grid, weights,
parameterization, bounds, and optimizer.

\begin{table}[H]
\centering
\caption{Calibration setup. Moneyness is measured in pilot standard deviations about the forward.}
\label{tab:calibration-setup}
\small
\begin{tabular}{@{}lp{0.69\linewidth}@{}}
\toprule
Item & Specification \\
\midrule
Smiles & 1Y$\times$5Y, 3Y$\times$7Y, and 5Y$\times$10Y (three smiles) \\
Fit grid & $[-1.25, -0.625, 0, 0.625, 1.25]$; five strikes per smile, 15 quotes \\
Holdout grid & $[-1, -0.3125, 0.3125, 1]$; four strikes per smile, 12 quotes \\
Weights & Normal-vega residuals with $\widetilde V=\max\{V,0.20V_{\rm ATM}\}$ \\
Coordinates & $\log(\alpha/0.012)$, $\beta/0.25$, and $\log(\epsilon/0.25)$, bucket by bucket \\
Bounds & $\alpha\in[0.004,0.025]$, $\beta\in[-0.45,0.10]$, $\epsilon\in[0.03,0.70]$ \\
Optimizer & Bounded least squares; two-point Jacobian, step $5\times10^{-4}$, Jacobian scaling, tolerance $2\times10^{-6}$, two starts \\
\bottomrule
\end{tabular}
\end{table}

Reference quotes for the ordinary regimes are generated independently of both calibration models. The estimator combines a frozen-loading PDE control variate with a common-path randomized-QMC estimate of the price difference between the nonlinear and frozen loadings, together with an independently randomized Richardson correction for time discretization. The one-bucket surfaces use eight independent scrambles, while the two-bucket surface uses 24. The corresponding maximum quote standard errors are \(\ModerateMaximumQuoteSeBp\), \(\NegativeMaximumQuoteSeBp\), and \(\TwoBucketMaximumQuoteSeBp\) normal-volatility bp. Calibration and held-out strikes are generated jointly from the same underlying randomized-QMC calculations. Each calibration is performed from two initial parameter vectors using common parameter bounds and convergence tolerances, and the full weighted Jacobian is used in the subsequent conditioning analysis.

The bounded least-squares coordinates are, bucket by bucket,
\(\log(\alpha/0.012)\), \(\beta/0.25\), and
\(\log(\epsilon/0.25)\). Their natural-parameter bounds are
\(\alpha\in[0.004,0.025]\), \(\beta\in[-0.45,0.10]\), and
\(\epsilon\in[0.03,0.70]\). These bounds apply to every regime.

Calibration uncertainty is examined using two paired replicate designs. First, \(\QmcReplicateCount\) quote replicates are constructed from independent randomized-QMC scrambles using the same estimator as the reference surface. Second, \(\MarketReplicateCount\) synthetic quote perturbations are generated with marginal scale \(\MarketQuoteNoiseBp\) normal-volatility bp and correlated level, expiry, skew, curvature, and idiosyncratic components. Antithetic pairing ensures that the perturbations have exactly zero mean across the ensemble. These perturbations assess sensitivity to small structured changes in the quote surface. In the singular-spectrum analysis, a change is classified as material only if it exceeds the 20\% ratio threshold and the corresponding paired confidence interval excludes one. Changes in the associated singular vectors are assessed separately.

The parameter regimes are summarized in \cref{tab:calibration_regimes}. The Moderate regime is the least stressed stochastic volatility specification, while the Negative skew regime introduces stronger negative dependence between rate and volatility shocks. The State-dependence stress combines slower rate mean reversion, larger rate volatility, stronger correlated volatility, and weaker volatility mean reversion. These changes increase the variation of the swap-rate loading over the relevant state region and therefore provide a more demanding test of the state-linear approximation. The Broad-volatility stress uses a separate parameter set with a larger volatility-diffusion scale. The Two buckets specification allows each of \(\alpha\), \(\beta\), and \(\epsilon\) to vary across two time intervals, with a common bucket boundary at five years.

\begin{table}[H]
\centering
  \caption{Synthetic calibration regimes. In all cases, \(\theta=\sigma_0=1\). For the Two buckets regime, parenthesized values are the first and second time buckets.}
\label{tab:calibration_regimes}
\small
\begin{tabular}{lrrrrrr}
\toprule
Regime & $\lambda$ & $\kappa_1$ & $\kappa_2$ & $\alpha$ & $\beta$ & $\epsilon$ \\
\midrule
Moderate & 0.060 & 0.45 & 0.65 & 0.0120 & -0.050 & 0.180 \\
Negative skew & 0.060 & 0.25 & 0.40 & 0.0120 & -0.200 & 0.300 \\
State-dependence stress & 0.025 & 0.10 & 0.25 & 0.0180 & -0.300 & 0.350 \\
Broad-volatility stress & 0.040 & 0.08 & 0.18 & 0.0140 & -0.380 & 0.550 \\
Two buckets & 0.060 & 0.25 & 0.40 & (0.0105, 0.0135) & (-0.100, -0.220) & (0.180, 0.300) \\
\bottomrule
\end{tabular}

\end{table}

For the common-random-number error decomposition, we use eight independently scrambled antithetic Sobol replicates with \(32768\) paths and \(128\) time steps per year. Time-discretization corrections are obtained from matched calculations at \(64\) and \(128\) time steps per year.

\subsection{First-variation transform validation}\label{sec:transform_design}

The transform validation uses the Negative skew parameter set in \cref{tab:calibration_regimes}. The principal cases are a \(3\mathrm{Y}\times7\mathrm{Y}\) and a \(10\mathrm{Y}\times10\mathrm{Y}\) swaption. The E1 transform stress uses the State-dependence stress parameters and a \(5\mathrm{Y}\times10\mathrm{Y}\) swaption. Deterministic coefficient paths are prepared on 24 calendar steps per year.

Three independent calculations of the first variation are compared.

\paragraph{One-dimensional PDE.}
The exact system \cref{eq:logvol_f0,eq:logvol_gx,eq:logvol_gy,eq:logvol_g0} is solved by Crank--Nicolson time stepping on an equally spaced log-volatility grid with centered finite differences and remote reflecting boundaries \citep{CrankNicolson1947}. The Negative skew reference uses \(241\) spatial points on \([-3.5,1.7]\) and 48 steps per year. Wider domains, \(321\) points, and 96 steps per year are used for refinement. 

\paragraph{Finite E1 ODE.}
The baseline coefficients \((A_0,A_1,A_2)\) and the projected correction coefficients in \cref{eq:block_ode} are solved with adaptive SciPy integration. Degrees \(N=3,4,6,8,10\) are used to diagnose stability; \(N=6\) is the principal specification for the Negative skew regime. The literal first-order reconstruction \cref{eq:mgf_reconstruction} is used.

\paragraph{Pathwise quasi-Monte Carlo.}
Under the frozen annuity drift, the state process does not depend on \(\gamma\), and the state-linear swap-rate increment is pathwise affine in \(\gamma\):
\begin{equation}
  \Delta S^\gamma_{0,T}=\Delta S^0_{0,T}+\gamma\mathcal I_{0,T},
  \qquad
  \mathcal I_{0,T}=\int_0^T\sigma_t(a_x\xi_t+a_y\eta_t)\dd W_t.
  \label{eq:pathwise_affine_gamma}
\end{equation}
Consequently,
\begin{equation}
  \left.\partial_\gamma
    \E[e^{-\phi\Delta S^\gamma_{0,T}}]\right|_{\gamma=0}
  =\E[-\phi\mathcal I_{0,T} e^{-\phi\Delta S^0_{0,T}}],
  \label{eq:pathwise_mgf_derivative}
\end{equation}
and for a normalized payer call,
\begin{equation}
  \left.\partial_\gamma
    \E[(F_0+\Delta S^\gamma_{0,T}-K)^+]\right|_{\gamma=0}
  =\E[\mathcal I_{0,T}\,\mathbf 1_{\{F_0+\Delta S^0_{0,T}>K\}}].
  \label{eq:pathwise_call_derivative}
\end{equation}
These identities avoid finite differencing in \(\gamma\). The Negative skew comparisons use eight independently scrambled Sobol replicates of \(8192\) paths; six replicates are used in the E1 transform stress \citep{Owen1997,Glasserman2004}.

The transform derivative is inverted with a damped double-exponential Fourier rule of the type discussed by \citet{Lewis2000,AndersenLake2019}. The implementation is separately checked against Bachelier prices from a Gaussian moment-generating function. Nine strikes from minus two to plus two terminal standard deviations are used in the price-level comparison.

\subsection{Error decomposition}
\label{sec:pricing_decomposition}

\begin{table}[t]
\centering
\caption{Pricing quantities used in the error decomposition. The quadratic-loading diagnostic $P_{\rm quad}$ is not an A--G label.}
\label{tab:pricing-decomposition}
\scriptsize
\resizebox{\textwidth}{!}{%
\begin{tabular}{@{}clllp{0.37\linewidth}@{}}
\toprule
Object & loading & drift & solver & role \\
\midrule
$A$ & frozen & frozen & Monte Carlo & constant-loading frozen-drift price \\
$B$ & state-linear & frozen & Monte Carlo & linear-loading frozen-drift price \\
$C$ & exact nonlinear & frozen & Monte Carlo & exact-loading frozen-drift price \\
$D$ & exact nonlinear & nonlinear & Monte Carlo & exact-loading exact-drift benchmark price \\
$P_{\rm quad}$ & quadratic & frozen & Monte Carlo & standalone quadratic-loading Monte Carlo diagnostic; not an A--G label \\
$E$ & frozen & frozen & E1 Fourier & zeroth-order Fourier price \\
$F$ & first-order corrected & frozen & E1 Fourier & first-order corrected Fourier price \\
$G_0$ & frozen & frozen & log-vol PDE/spectral & zeroth-order PDE price \\
$G$ & first-order corrected & frozen & log-vol PDE/spectral & first-order corrected PDE price \\
\bottomrule
\end{tabular}}
\end{table}

We use \(B-A\) for the state-linear loading effect, \(C-B\) for the loading
Taylor remainder, and \(D-C\) for the frozen-drift error. The E1 total
first-order residual is \((F-E)-(B-A)\), and the PDE/spectral total residual
is \((G-G_0)-(B-A)\). These quantities include both the first-order
truncation remainder at \(\gamma=1\) and any error in the corresponding
numerical representation of the first variation.

\subsection{Accuracy criteria}
\label{sec:validity_regions}

We fixed the following error thresholds before evaluating the stress cases. Let \(\lVert\cdot\rVert\) denote the RMS norm over the complete expiry--tenor--strike panel, and let
\[
    L=\lVert B-A\rVert
\]
denote the price effect of replacing the frozen loading by its state-linear approximation in \cref{tab:pricing-decomposition}. All relative errors are measured against this quantity. Let \(R_{\mathrm{total}}\) denote the total first-order residual for the numerical representation used in a given regime.

For a given parameter regime, we require the numerical representation used for evaluation to be stable and to satisfy the price-admissibility checks, and the state-linear loading effect \(L\) to be statistically distinguishable from simulation error. The approximation errors must additionally satisfy
\begin{equation}
  \frac{\lVert C-B\rVert}{L}\leq 10\%,\qquad
  \frac{\lVert D-C\rVert}{L}\leq 5\%,\qquad
  \frac{\lVert R_{\gamma=1}\rVert}{L}\leq 10\%,\qquad
  \frac{\lVert R_{\mathrm{total}}\rVert}{L}\leq 5\%.
  \label{eq:credible_region}
\end{equation}
The maximum absolute value of \(R_{\mathrm{total}}\) at any individual quote must not exceed \(\QuoteResidualThresholdBp\) normal-volatility bp. As an additional diagnostic of the loading expansion, the quadratic approximation is considered accurate when
\[
    \frac{\lVert P_{\mathrm{ex}}-P_{\mathrm{quad}}\rVert}
         {\lVert P_{\mathrm{ex}}-P_{\mathrm{lin}}\rVert}\leq 10\%.
\]
Failure of the E1 projection need not imply failure of the one-dimensional first-variation equations. We use E1 when degree and time-discretization refinements indicate numerical stability; otherwise, we solve the same equations with the log-volatility PDE/spectral method, subject to the same refinement and price-admissibility requirements.

The ordinary regimes meet these criteria with E1. In both stress regimes, the PDE/spectral calculation is numerically stable but the total first-order residual exceeds the tolerance. This classification applies only to the five parameter specifications considered here.

\subsection{Robustness checks}
\label{sec:robustness_design}

We add two robustness checks: an interpolation toward the State-dependence stress and a non-flat initial curve.

First, we consider a one-dimensional family of parameter specifications connecting the Negative skew regime to the State-dependence stress:
\[
  \vartheta(u)
  =
  (1-u)\vartheta_{\mathrm{Negative\ skew}}
  +
  u\vartheta_{\mathrm{State\ stress}},
  \qquad
  u\in[0,1].
\]
Rate mean reversion, the two volatility-drift coefficients, long-run and initial volatility, and the parameters \(\alpha\), \(\beta\), and \(\epsilon\) are interpolated linearly in their natural parameterization. Across \(u\), we use the same swaption panel, randomized-QMC construction, randomizations, accuracy thresholds, and rule for choosing between E1 and the PDE/spectral representation. The resulting path shows how the approximation changes between the Negative skew and State-dependence stress specifications.
Second, we examine sensitivity to the initial term structure. The flat \(2.5\%\) initial curve is replaced by a non-flat continuously compounded zero curve with nodes
\[
  \begin{aligned}
  (T,z(T))={}&(0,1.80\%),\,(1,2.00\%),\,(2,2.15\%),\,(5,2.40\%),\\
             &(10,2.70\%),\,(20,2.90\%),\,(30,3.00\%).
  \end{aligned}
\]
Zero rates are linearly interpolated between these maturities and held constant beyond the endpoint nodes. We chose the curve and interpolation rule before running the comparison. All other settings match the Moderate parameter regime.

\section{Numerical results}\label{sec:numerical_results}

\subsection{Accuracy of the loading approximation}
\label{sec:loading_approximation_accuracy}

The identity \(E_x+E_y+E_{xy}=1\) holds to numerical precision in every row.
The table reports the cross term explicitly, so \(E_y\) is not an additive
share.

\begin{samepage}
\noindent The state-linear loading error is largest in the two stress cases,
although the quadratic Taylor term substantially reduces it.
\end{samepage}

\begin{table}[t]
\centering
  \caption{Accuracy of the loading approximation. Entries after $\pm$ are randomized-QMC standard errors. Maxima are over the three calibrated instruments; the instrument column identifies their common maximizer and separately notes any different $E_{xy}$ maximizer. Each estimate uses 8 antithetic scrambled-Sobol replicates of 32768 paths and a first-order Richardson correction from paired 64- and 128-step-per-year grids. Reduction denotes $R_{\rm fr}/R_{\rm lin}$; the price diagnostic is $\lVert P_{\rm ex}-P_{\rm quad}\rVert/\lVert P_{\rm ex}-P_{\rm lin}\rVert$.}
\label{tab:loading-approximation-accuracy}
\footnotesize
\begin{tabular}{@{}llcccc@{}}
\toprule
Regime & instrument(s) & $R_{\rm fr}$ & $R_{\rm lin}$ & $R_{\rm quad}$ & reduction \\
\midrule
Moderate & 5Yx10Y & 0.02491$\pm$1e-05 & 0.0007205$\pm$1e-06 & 1.362e-05$\pm$5e-08 & 34.6 \\
Negative skew & 5Yx10Y & 0.0289$\pm$5e-05 & 0.001254$\pm$2e-05 & 2.749e-05$\pm$8e-07 & 23.0 \\
Two buckets & 5Yx10Y & 0.02765$\pm$2e-05 & 0.00108$\pm$8e-06 & 2.284e-05$\pm$3e-07 & 25.6 \\
State-dependence stress & 5Yx10Y & 0.05923$\pm$0.0001 & 0.008208$\pm$0.0002 & 0.0003688$\pm$3e-05 & 7.2 \\
Broad-volatility stress & 5Yx10Y & 0.07869$\pm$0.0004 & 0.01711$\pm$0.0003 & 0.0006239$\pm$4e-05 & 4.6 \\
\bottomrule
\end{tabular}

\vspace{4pt}

\begin{tabular}{@{}llccc@{}}
\toprule
Regime & instrument(s) & max $E_y$ & $E_{xy}$ at max $|E_{xy}|$ & price diagnostic \\
\midrule
Moderate & 5Yx10Y & 3.151e-05$\pm$4e-08 & -0.00285$\pm$9e-06 & 0.96\% \\
Negative skew & 5Yx10Y & 0.0002726$\pm$7e-07 & -0.01934$\pm$4e-05 & 0.95\% \\
Two buckets & 5Yx10Y & 7.115e-05$\pm$2e-07 & -0.009404$\pm$3e-05 & 1.03\% \\
State-dependence stress & $E_y$: 5Yx10Y; $E_{xy}$: 3Yx7Y & 0.0006169$\pm$4e-06 & -0.02489$\pm$0.0002 & 2.34\% \\
Broad-volatility stress & 5Yx10Y & 0.003596$\pm$7e-05 & -0.09078$\pm$0.001 & 1.60\% \\
\bottomrule
\end{tabular}
\end{table}

\subsection{Calibration accuracy and parameter recovery}
\label{sec:coefficient_calibration_results}

\Cref{tab:main_calibration_results} summarizes the calibration results for the ordinary regimes. The reference swaption prices are generated independently under the nonlinear swap-rate loading and are therefore not produced by either of the models used in calibration. In all three cases, incorporating the first-order state dependence of the loading reduces both parameter-recovery bias and pricing errors at the held-out strikes. The improvement is most pronounced in the one-bucket specifications and remains visible, although smaller, in the six-parameter Two buckets calibration.

\begin{table}[H]
  \centering
  \caption{Calibration to independently generated nonlinear-model swaption prices. Fit and holdout RMSE are reported in normal-volatility basis points. The \(\beta\)-bias column reports the largest absolute bias across time buckets.}
  \label{tab:main_calibration_results}
  \small
  \resizebox{\textwidth}{!}{\begin{tabular}{llrrrr}
\toprule
Regime & Pricer & Fit RMSE & Holdout RMSE & Max. $|\beta|$ bias & Condition \\
\midrule
Moderate & Frozen E1 & 0.3771 & 0.3169 & 0.03005 & 326.4 \\
Moderate & Corrected E1 & 0.0151 & 0.0118 & 0.00006 & 271.4 \\
\addlinespace
Negative skew & Frozen E1 & 0.3136 & 0.2622 & 0.02793 & 61.1 \\
Negative skew & Corrected E1 & 0.0381 & 0.0323 & 0.00342 & 59.9 \\
\addlinespace
Two buckets & Frozen E1 & 0.1236 & 0.1102 & 0.06336 & 344.2 \\
Two buckets & Corrected E1 & 0.0668 & 0.0658 & 0.00640 & 315.8 \\
\bottomrule
\end{tabular}
}
\end{table}

For the Moderate specification, the held-out RMSE decreases from
\(\ModerateFrozenHeldoutBp\) to \(\ModerateCorrectedHeldoutBp\) normal-volatility bp. Under the Negative skew specification it decreases from
\(\NegativeFrozenHeldoutBp\) to \(\NegativeCorrectedHeldoutBp\) bp, while in the Two buckets specification it decreases from
\(\TwoBucketFrozenHeldoutBp\) to \(\TwoBucketCorrectedHeldoutBp\) bp. These correspond to reductions by factors of
\(\ModerateHoldoutImprovementFactor\),
\(\NegativeHoldoutImprovementFactor\), and
\(\TwoBucketHoldoutImprovementFactor\), respectively. Both initializations converge to interior solutions in each case, and the sampling uncertainty of the reference quotes is small relative to the reported pricing differences.

\begin{figure}[H]
  \centering
  \includegraphics[width=0.74\textwidth]{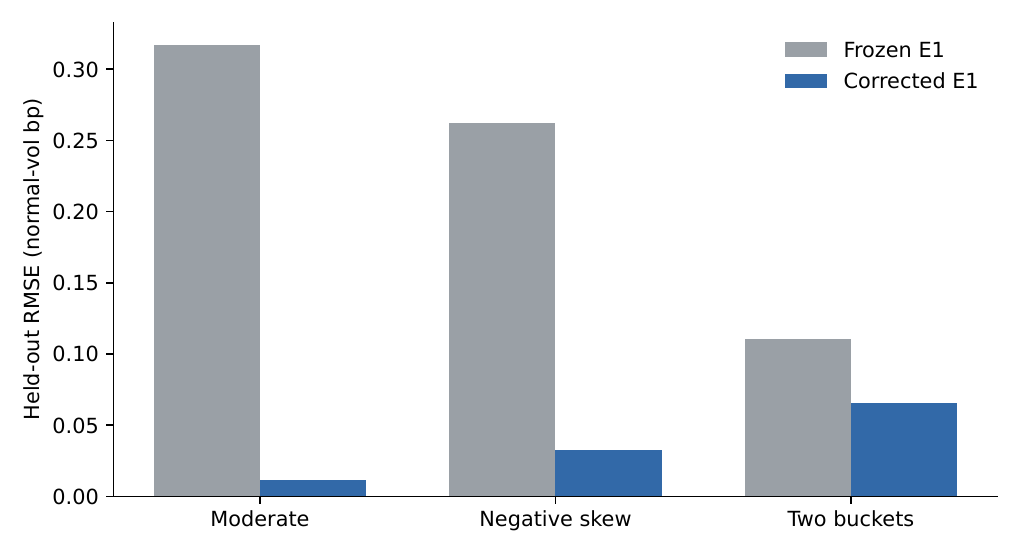}
  \caption{Held-out smile RMSE for the frozen and state-linear E1 calibrations. The state-linear correction reduces the pricing error in each of the three ordinary regimes.}
  \label{fig:holdout_rmse}
\end{figure}

The effect on parameter recovery is particularly clear for \(\beta\), which determines the component of the volatility shock correlated with the curve-factor Brownian motion. \Cref{tab:beta_recovery} reports the fitted values. Under the frozen-loading approximation, the estimated magnitude of \(\beta\) is systematically biased toward zero. Incorporating the state-linear loading substantially reduces this bias in every bucket, without introducing additional calibration parameters.

\begin{table}[H]
  \centering
  \caption{Recovery of the correlated volatility loading.}
  \label{tab:beta_recovery}
  \small
  \begin{tabular}{lrrrr}
\toprule
Regime & Bucket & Truth & Frozen E1 & Corrected E1 \\
\midrule
Moderate & 1 & -0.0500 & -0.0199 & -0.0499 \\
Negative skew & 1 & -0.2000 & -0.1721 & -0.1966 \\
Two buckets & 1 & -0.1000 & -0.0823 & -0.0995 \\
Two buckets & 2 & -0.2200 & -0.1566 & -0.2136 \\
\bottomrule
\end{tabular}

\end{table}

\Cref{tab:full-model-repricing} compares the nonlinear-model prices implied by the
parameter estimates obtained from the frozen and corrected calibrations. For each
regime, both parameter vectors are therefore evaluated under the same full
nonlinear dynamics. The parameters obtained from the corrected calibration
reproduce the nonlinear data-generating price surface substantially more closely,
both at the calibration strikes and at the held-out strikes. This improvement is
present in all three regimes, and the paired confidence intervals for the
reductions in RMSE exclude zero.

\begin{table}[H]

  \centering

  \caption{Full nonlinear-model repricing under the parameter estimates obtained
  from the frozen and corrected calibrations. Errors are measured relative to the
  nonlinear data-generating surface and reported in normal-volatility basis
  points. The final column reports the maximum absolute error across the fitting
  and held-out panels. The largest quote-level Monte Carlo standard error is
  \(\FullModelRepricingMaximumSeBp\) bp.}

  \label{tab:full-model-repricing}

  \small

  \begin{tabular}{llrrr}
\toprule
Regime & Parameters from & NL fit RMSE & NL holdout RMSE & Max. NL error \\
\midrule
Moderate & Frozen calibration & 0.8688 & 0.7297 & 1.3737 \\
Moderate & Corrected calibration & 0.0170 & 0.0178 & 0.0211 \\
\addlinespace
Negative skew & Frozen calibration & 0.9996 & 0.8526 & 1.5547 \\
Negative skew & Corrected calibration & 0.1125 & 0.0973 & 0.2003 \\
\addlinespace
Two buckets & Frozen calibration & 0.9112 & 0.7715 & 2.0363 \\
Two buckets & Corrected calibration & 0.0943 & 0.0898 & 0.2199 \\
\bottomrule
\end{tabular}

\end{table}

\subsubsection{Calibration uncertainty and local identifiability}

Parameter standard deviations are nearly unchanged after introducing the
correction. In both replicate designs, the paired 95\% confidence interval for
corrected-minus-frozen holdout RMSE lies below zero.

\begin{table}[!t]
\centering
\caption{Calibration uncertainty. Parameter entries are the largest bucketwise standard deviations for frozen/corrected fits. Intervals are 95\% paired percentile-bootstrap intervals for the median.}
\label{tab:calibration-uncertainty-summary}
\footnotesize
\begin{tabular}{@{}llrcc@{}}
\toprule
Regime & family & $n$ & sd$(\beta)$ F/C & sd$(\epsilon)$ F/C \\
\midrule
Moderate & QMC & 8 & 0.00013/0.00013 & 0.0011/0.001 \\
Moderate & Quote perturbation & 20 & 0.0028/0.0028 & 0.038/0.033 \\
Negative skew & QMC & 8 & 0.00012/0.00012 & 0.00038/0.00037 \\
Negative skew & Quote perturbation & 20 & 0.0023/0.0023 & 0.0093/0.0092 \\
Two buckets & QMC & 8 & 0.00021/0.00021 & 0.0019/0.0019 \\
Two buckets & Quote perturbation & 20 & 0.0059/0.0059 & 0.079/0.073 \\
\bottomrule
\end{tabular}

\vspace{4pt}

\begin{tabular}{@{}llccc@{}}
\toprule
Regime & family & $\Delta$ holdout RMSE (bp) & $s_{\min}^{\rm C}/s_{\min}^{\rm F}$ & $|v_{\min}^{\rm C\top}v_{\min}^{\rm F}|$ \\
\midrule
Moderate & QMC & -0.304 [-0.307,-0.299] & 1.204 [1.202,1.205] & 0.999999 \\
Moderate & Quote perturbation & -0.199 [-0.210,-0.151] & 1.204 [1.167,1.269] & 0.999999 \\
Negative skew & QMC & -0.227 [-0.235,-0.224] & 1.020 [1.020,1.020] & 0.999987 \\
Negative skew & Quote perturbation & -0.136 [-0.162,-0.121] & 1.020 [1.019,1.022] & 0.999987 \\
Two buckets & QMC & -0.045 [-0.045,-0.044] & 1.091 [1.089,1.092] & 0.999736 \\
Two buckets & Quote perturbation & -0.019 [-0.025,-0.007] & 1.095 [1.083,1.121] & 0.999790 \\
\bottomrule
\end{tabular}
\end{table}

The singular-value spectra of the weighted calibration Jacobians measure local parameter identifiability using the same floored normal-vega residuals and transformed parameter coordinates as calibration. \Cref{fig:singular_spectra} averages singular values across QMC replicates and normalizes by the first component, whereas \cref{tab:calibration-uncertainty-summary} reports paired corrected-to-frozen ratios.

\begin{figure}[H]
  \centering
  \includegraphics[width=0.92\textwidth]{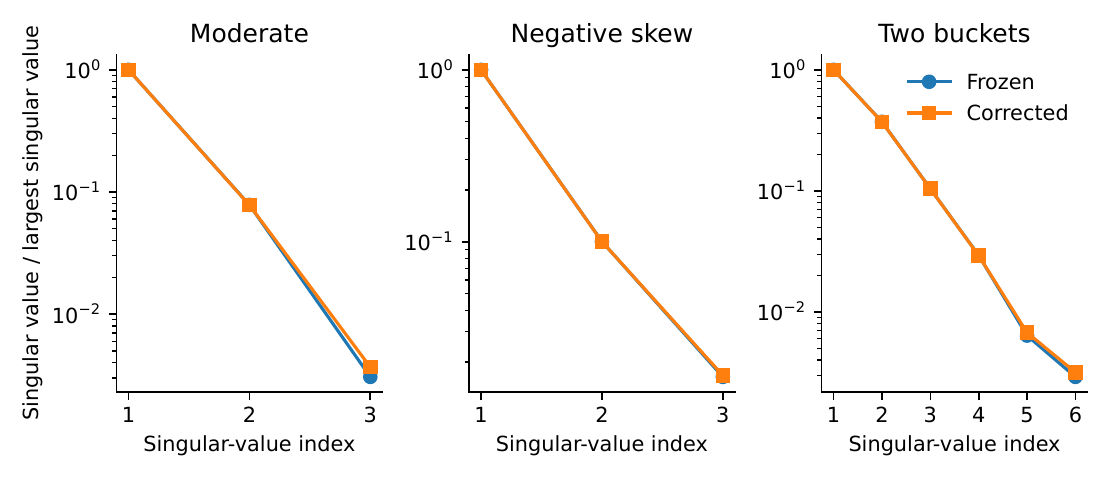}
  \caption{Normalized singular-value spectra of the weighted calibration Jacobians for the QMC replicate experiments. The state-linear correction has substantially less effect on the weakly identified parameter directions than on parameter-recovery bias.}
  \label{fig:singular_spectra}
\end{figure}

The effect of the correction on local calibration identifiability is limited. In the Moderate specification, the smallest singular value increases by \(\ModerateWeakestLiftPercent\)\%, while its associated right-singular vector is essentially unchanged. The corresponding changes are smaller in the Negative skew and Two buckets specifications. More generally, the weakest singular directions are highly stable across the frozen and corrected calibrations, with cosines exceeding \(0.999\). They remain dominated by \(\log\epsilon\) and, in the Two buckets specification, by contrasts between the \(\epsilon\) buckets. The correction therefore improves parameter recovery primarily by reducing approximation-induced bias rather than by materially altering the weakly identified directions of the calibration problem.

\subsubsection{PDE-based calibration in stressed regimes}

For the stress regimes, the finite E1 representation is replaced by the log-volatility PDE/spectral implementation once the E1 approximation ceases to exhibit satisfactory numerical stability. In the State-dependence stress, the PDE-based calibration attains an in-sample RMSE of \(\StateStressPdeFitBp\) normal-volatility bp and a held-out RMSE of \(\StateStressPdeHeldoutBp\) bp. In the Broad-volatility stress, the corresponding errors are \(\BroadStressPdeFitBp\) and \(\BroadStressPdeHeldoutBp\) bp.

In both stress cases, the calibrated solutions are numerically stable, insensitive to initialization, and interior to the admissible parameter domain. The remaining calibration error therefore reflects the first-order model approximation. The maximum quote-level residual nevertheless exceeds the prescribed tolerance, placing both specifications outside the range of quantitative accuracy considered here.

\subsubsection{Decomposition of approximation errors}

To quantify the contributions of the successive approximation layers, we use
the matched, common-random-number objects in
\cref{tab:pricing-decomposition}. The E1 total first-order residual is
\((F-E)-(B-A)\); for the PDE route it is
\((G-G_0)-(B-A)\). These aligned increments remove any common level
discrepancy between the Monte Carlo and analytical implementations from the
measured total correction error. The quadratic diagnostic compares
\(P_{\mathrm{quad}}\) with \(P_{\mathrm{ex}}\) and \(P_{\mathrm{lin}}\).

\begin{table}[H]
  \centering
  \caption{Decomposition of approximation errors across parameter regimes. RMS error components are reported relative to \(L=\lVert B-A\rVert\), the RMS pricing effect of the state-linear loading. The final column reports the maximum absolute quote-level total first-order residual in normal-volatility basis points. \(R_{\gamma=1}\) denotes the first-order truncation remainder at \(\gamma=1\), and \(R_{\mathrm{total}}\) the total residual of the E1 or PDE/spectral first-order correction used for that regime.}
  \label{tab:credible_decomposition}
  \scriptsize
  \resizebox{\textwidth}{!}{\begin{tabular}{llrrrrr}
\toprule
Regime & Representation & $\|C-B\|/L$ & $\|D-C\|/L$ & $\|R_{\gamma=1}\|/L$ & $\|R_{\rm total}\|/L$ & $\max|R_{\rm total}|$ \\
\midrule
Moderate & E1 & 1.83\% & 0.02\% & 0.47\% & 0.51\% & 0.0133 \\
Negative skew & E1 & 1.99\% & 0.37\% & 0.60\% & 0.65\% & 0.0158 \\
Two buckets & E1 & 2.02\% & 0.35\% & 0.56\% & 0.63\% & 0.0152 \\
State-dependence stress & PDE & 3.77\% & 2.09\% & 1.31\% & 1.35\% & 0.0981 \\
Broad-volatility stress & PDE & 4.01\% & 1.96\% & 2.05\% & 2.18\% & 0.0738 \\
\bottomrule
\end{tabular}
}
\end{table}

For the ordinary regimes, all error components remain within the accuracy criteria of \cref{sec:validity_regions}. The two stress regimes also satisfy the relative error criteria, the quadratic-loading diagnostic, and the numerical-refinement and price-admissibility requirements. Their largest quote-level total first-order residuals, however, are \(\StateStressResidualBp\) and \(\BroadStressResidualBp\) normal-volatility bp, respectively, exceeding the tolerance of \(\QuoteResidualThresholdBp\) bp. The first-order approximation is therefore accurate on an RMS-relative basis in these cases but not uniformly across individual quotes.

The order of the first-order truncation error is examined separately by varying the loading scale \(\gamma\). For each regime and instrument, we evaluate the remainder after subtracting the linear term and estimate its slope with respect to \(\gamma\) on logarithmic scales. As reported in \cref{tab:gamma_scaling_results}, the estimated slopes lie between \(\GammaSlopeMinimum\) and \(\GammaSlopeMaximum\), including in the stressed specifications, consistent with the quadratic remainder expected from the first-order expansion.

\begin{table}[H]
  \centering
  \caption{Log--log slopes of the first-order transform remainder as a function of the loading scale \(\gamma\). A slope of two corresponds to the quadratic remainder predicted by the perturbation expansion.}
  \label{tab:gamma_scaling_results}
  \small
  \begin{tabular}{lrrr}
\toprule
Regime & Minimum slope & Maximum slope & Resolved instruments \\
\midrule
Moderate & 1.974 & 1.986 & 3/3 \\
Negative skew & 2.016 & 2.050 & 3/3 \\
State-dependence stress & 2.007 & 2.040 & 3/3 \\
Two buckets & 1.994 & 2.013 & 3/3 \\
Broad-volatility stress & 1.977 & 2.071 & 3/3 \\
\bottomrule
\end{tabular}

\end{table}

\subsection{Approximation accuracy along the interpolation path}
\label{sec:boundary_results}

\Cref{fig:validity_boundary} reports the first-order accuracy diagnostics along the parameter interpolation defined in \cref{sec:robustness_design}. The finite E1 representation remains numerically stable up to the tested value \(u=\BoundaryEoneLastU\), while the log-volatility PDE/spectral representation is used from \(u=\BoundaryPdeFirstU\) onward. The transition between the two numerical representations is therefore localized only to the interval
\[
  u\in[\BoundaryEoneLastU,\BoundaryPdeFirstU].
\]

The maximum quote-level first-order residual is estimated as
\(\BoundaryLowerResidualBp\pm\BoundaryLowerResidualSeBp\) normal-volatility bp at \(u=\BoundaryLowerU\) and
\(\BoundaryUpperResidualBp\pm\BoundaryUpperResidualSeBp\) bp at \(u=\BoundaryUpperU\). The corresponding point estimates lie on opposite sides of the prescribed tolerance \(\QuoteResidualThresholdBp\) bp, while their one-standard-error intervals overlap the threshold. The uncertainty is too large to refine this bracket meaningfully.

\begin{figure}[!t]
  \centering
  \includegraphics[width=0.91\textwidth]{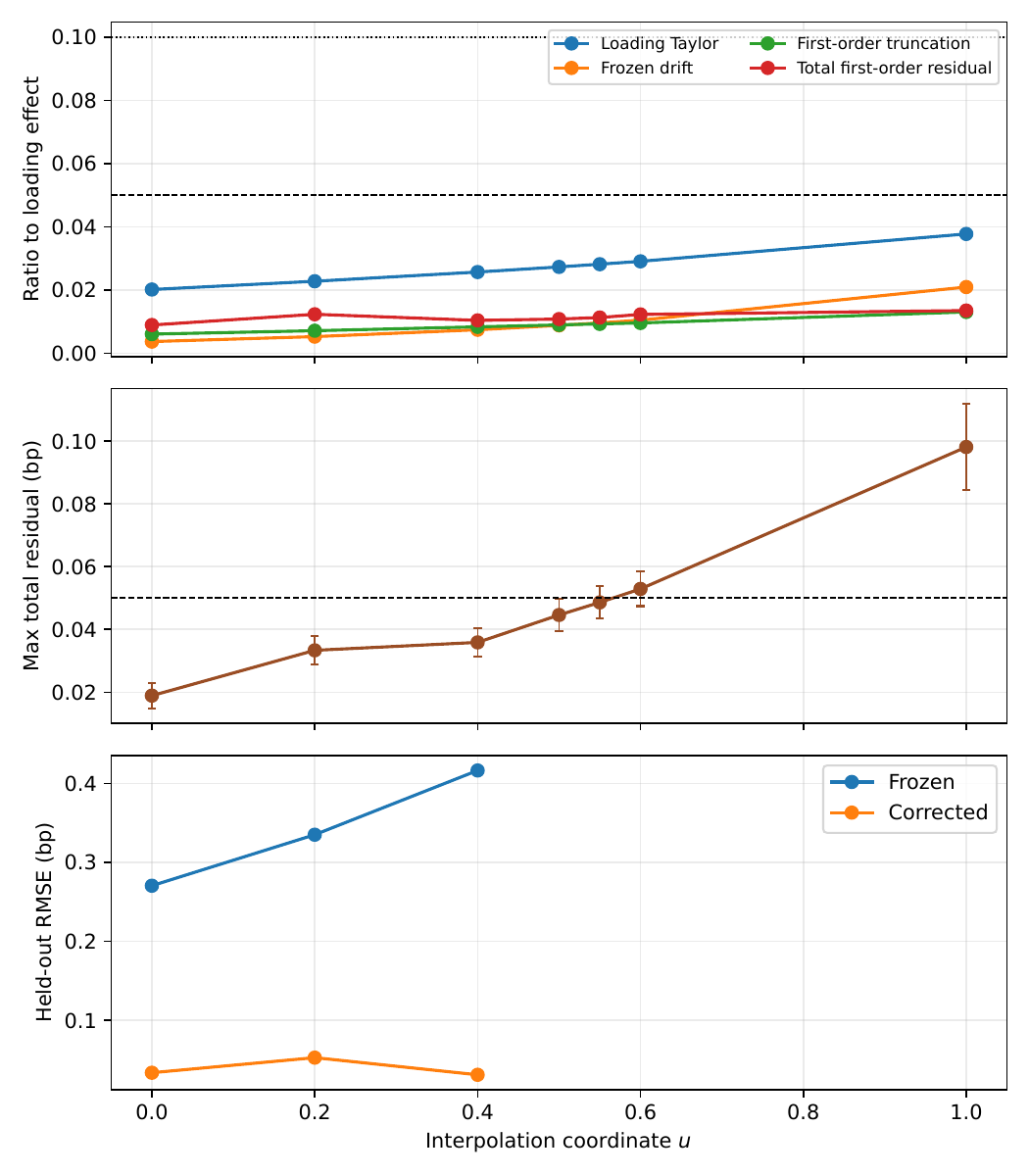}
  \caption{First-order approximation diagnostics along the interpolation from the Negative skew specification to the State-dependence stress. Horizontal lines indicate the accuracy tolerances. The interpolation locates a transition along one parameter path; it does not define a validity region in the full parameter space.}
  \label{fig:validity_boundary}
\end{figure}

\subsection{Sensitivity to the initial term structure}
\label{sec:nonflat_results}

We next examine whether the calibration results for the Moderate specification are sensitive to the assumption of a flat initial yield curve. Replacing the flat \(2.5\%\) curve by the upward-sloping term structure specified in \cref{sec:robustness_design} produces results consistent with those obtained under the baseline curve. The state-linear calibration continues to recover the data-generating stochastic volatility parameters closely, whereas the frozen-loading specification exhibits the same parameter-compensation effect observed in the baseline experiment.

\Cref{tab:nonflat_moderate} summarizes the calibration results. The held-out RMSE decreases from \(\NonflatFrozenHeldoutBp\) to \(\NonflatCorrectedHeldoutBp\) normal-volatility bp, corresponding to an improvement factor of \(\NonflatHoldoutImprovementFactor\). The data-generating instantaneous rate--volatility correlation is \(\NonflatTrueRho\), compared with calibrated values of \(\NonflatFrozenRho\) under the frozen specification and \(\NonflatCorrectedRho\) under the state-linear specification. Both initializations converge to interior solutions. The maximum analytical quote residual is \(\NonflatMaximumResidualBp\) bp, while the maximum standard error of the reference quotes is \(\NonflatMaximumQuoteSeBp\) bp; the numerical accuracy criteria for the E1 representation remain satisfied.

\begin{table}[H]
  \centering
  \caption{Calibration results for the Moderate specification under the upward-sloping initial yield curve. Pricing errors are reported in normal-volatility basis points.}
  \label{tab:nonflat_moderate}
  \small
  \resizebox{\textwidth}{!}{\begin{tabular}{lrrrrrr}
\toprule
Case & $\beta$ & $\epsilon$ & $\nu$ & $\rho$ & Fit RMSE & Holdout RMSE \\
\midrule
Truth & -0.050000 & 0.180000 & 0.186815 & -0.267644 & -- & -- \\
Frozen E1 & -0.020164 & 0.163033 & 0.164275 & -0.122744 & 0.3779 & 0.3174 \\
Corrected E1 & -0.049872 & 0.178252 & 0.185098 & -0.269436 & 0.0137 & 0.0104 \\
\bottomrule
\end{tabular}
}
\end{table}

\begin{minipage}{\textwidth}
For comparison, under the flat initial curve the corresponding held-out RMSEs are \(\ModerateFrozenHeldoutBp\) and \(\ModerateCorrectedHeldoutBp\) bp for the frozen and state-linear specifications, respectively. The reduction in calibration bias and held-out pricing error is therefore not specific to the flat-curve benchmark. The sensitivity of the results across a broader range of initial yield curves is left for further empirical investigation.
\end{minipage}

\begin{figure}[H]
  \centering
  \includegraphics[width=0.78\textwidth]{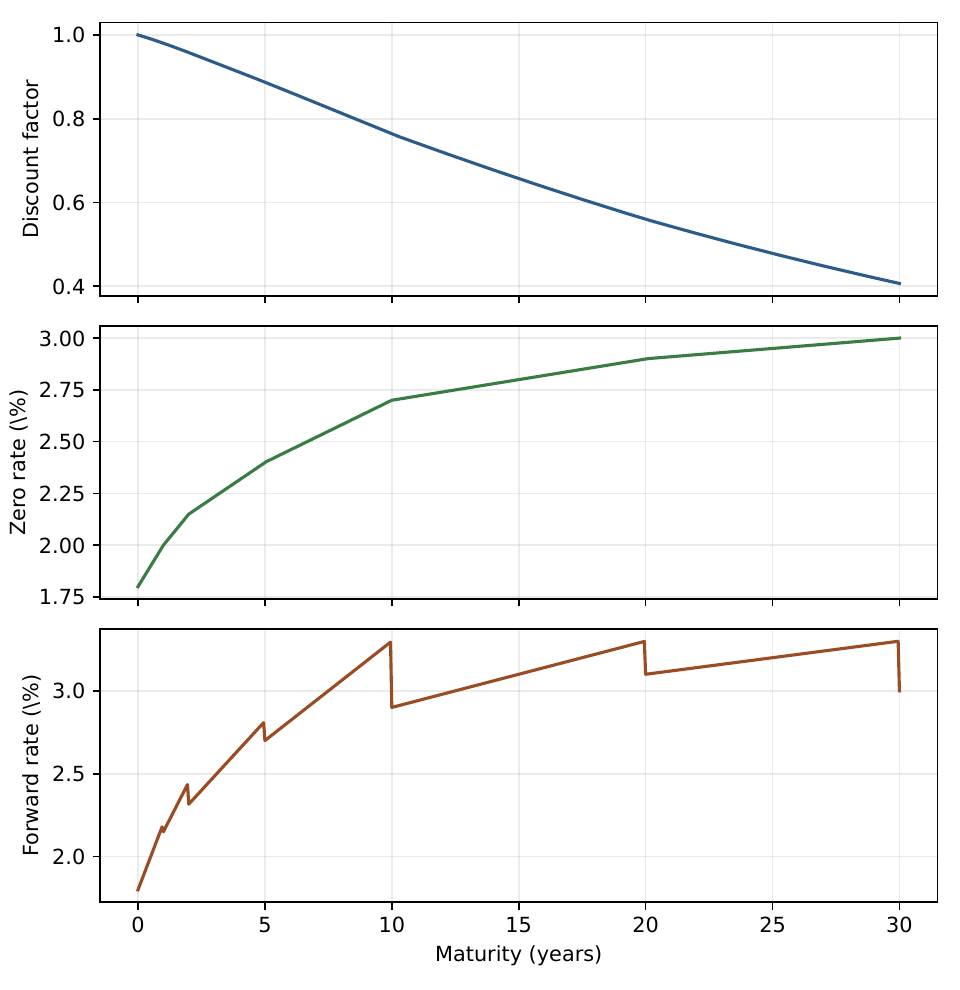}
  \caption{Upward-sloping synthetic initial term structure used in the sensitivity analysis, shown as discount factors, continuously compounded zero rates, and instantaneous forward rates.}
  \label{fig:nonflat_curve}
\end{figure}

\subsection{Exact first-variation reduction and symbolic checks}\label{sec:symbolic_results}

The symbolic calculation gives identically zero residuals for the constant, \(\xi\), and \(\eta\) coefficients in both the direct system \cref{eq:g_system} and the log-conjugated system \cref{eq:p_system}. It also verifies that no \(\xi^2\), \(\xi\eta\), or \(\eta^2\) coefficient is generated at first order. This confirms the invariant affine-state subspace used in \cref{prop:first_variation}.

Several identities provide additional consistency checks on the transform equations. At \(\phi=0\), the baseline transform reduces to one and the first variation vanishes, while along the Fourier contour the numerical solution satisfies the expected complex-conjugacy relation. Setting \(a_x=a_y=0\) recovers the frozen E1 transform to numerical precision.

As a more stringent check of the \(Y\)-dependent terms, we repeat the comparison after multiplying the full \(a_y(t)\) path by 20. At \(\phi=-0.5+15\ii\), the E1--PDE discrepancy remains \(4.44\times10^{-7}\), and the PDE result lies within approximately \(1.12\) scrambled-replicate standard errors of the corresponding pathwise Monte Carlo estimate. This comparison directly probes both the \(q_y\) source term and the \(p_x\)-to-\(p_y\) drift coupling.

The remaining numerical checks cover calendar-time reversal, the multidimensional matrix orientation, the Gaussian and zero-mean-reversion limits, pathwise derivatives, Fourier inversion, and call-price shape.

The additive first-order E1 reconstruction in \cref{eq:mgf_reconstruction} satisfies the numerical mass and option-price admissibility criteria used in the calculations, although Fourier inversion reveals a small negative component in the associated approximate density. This is consistent with its interpretation as a first-order perturbative approximation: the expansion does not guarantee positivity beyond first order. We use the additive reconstruction as the first-order formula throughout without interpreting it as an exact probability density. Exponentiating the correction would impose a different functional form and introduce higher-order terms that the expansion does not determine.

\subsection{Accuracy of the finite E1 correction in ordinary regimes}\label{sec:ordinary_transform_results}

\Cref{tab:transform_errors} compares the degree-six E1 log-transform correction with the exact one-dimensional PDE. The error is reported relative to the magnitude of the PDE correction, not relative to the complete transform.

\begin{table}[H]
  \centering
  \caption{Degree-six E1 correction versus the exact one-dimensional PDE in the Negative skew regime.}
  \label{tab:transform_errors}
  \small
  \begin{tabular}{lrrr}
\toprule
Swaption & Frequency & Absolute error & Relative correction error \\
\midrule
$\mathrm{3Yx7Y}$ & 0 & $7.93e-10$ & 0.233\% \\\\
$\mathrm{3Yx7Y}$ & 5 & $7.56e-08$ & 0.215\% \\\\
$\mathrm{3Yx7Y}$ & 15 & $5.19e-07$ & 0.145\% \\\\
$\mathrm{3Yx7Y}$ & 30 & $5.87e-06$ & 0.312\% \\\\
$\mathrm{3Yx7Y}$ & 60 & $7.22e-05$ & 0.662\% \\\\
\addlinespace
$\mathrm{10Yx10Y}$ & 5 & $1.72e-07$ & 0.151\% \\\\
$\mathrm{10Yx10Y}$ & 15 & $1.70e-06$ & 0.126\% \\\\
$\mathrm{10Yx10Y}$ & 30 & $5.40e-06$ & 0.066\% \\\\
$\mathrm{10Yx10Y}$ & 60 & $4.43e-05$ & 0.088\% \\\\
\bottomrule
\end{tabular}

\end{table}

\begin{figure}[H]
  \centering
  \includegraphics[width=0.78\textwidth]{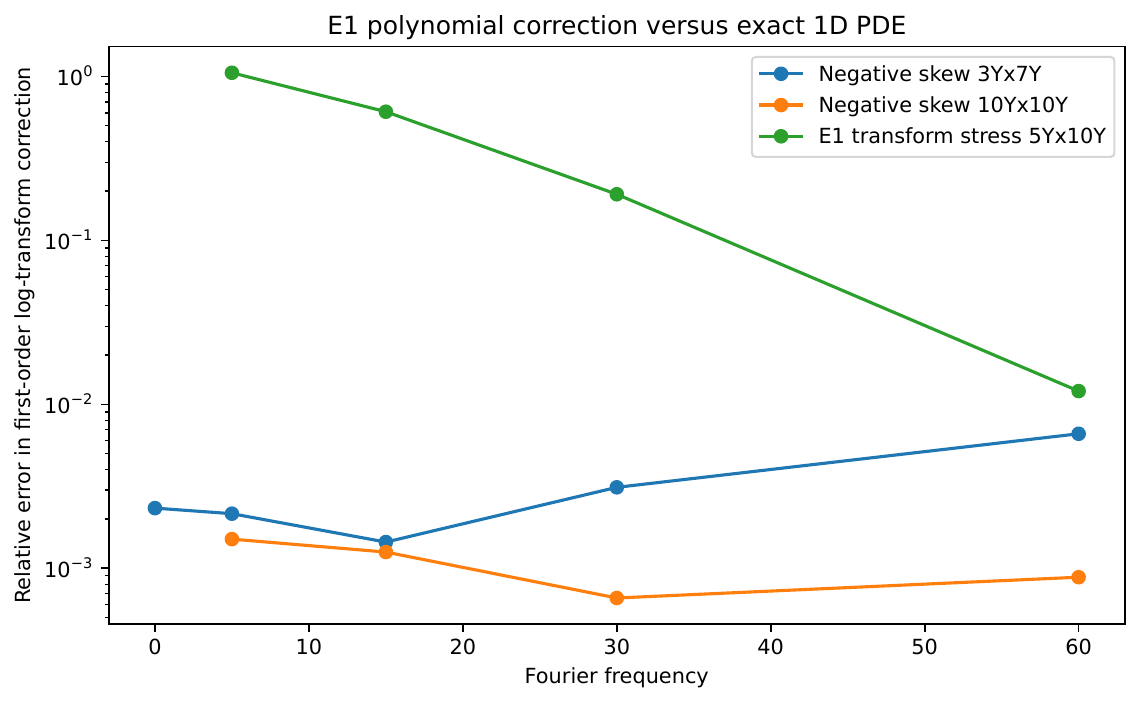}
  \caption{Relative error of the degree-six first-order correction against the exact one-dimensional PDE. The Negative skew cases are the principal validation cases; the E1 transform stress is discussed separately in \cref{sec:stress_results}.}
  \label{fig:transform_relative_error}
\end{figure}

The one-dimensional PDE is also checked by grid and domain refinement. At the \(10\mathrm{Y}\times10\mathrm{Y}\) swaption and Fourier frequency \(60\), the refined-grid solution changes the log-transform derivative by \(4.20\times10^{-7}\), approximately two orders of magnitude less than the corresponding E1 approximation error. The discrepancy between the PDE and pathwise Monte Carlo estimates decreases systematically as the Monte Carlo time grid is refined, consistent with time-discretization error in the simulated state dynamics.
\subsection{Price-level first-variation validation}\label{sec:price_validation_results}

\Cref{tab:price_derivative_results} compares Fourier inversion of the degree-six first-order transform with the exact pathwise option-price derivative in \cref{eq:pathwise_call_derivative}. The errors are annuity-normalized prices. At the at-the-money strike, they correspond to \(\ThreeBySevenAtmEquivalentBp\) and \(\TenByTenAtmEquivalentBp\) normal-volatility basis points for the two cases.

\begin{table}[H]
  \centering
  \caption{First-order payer-swaption price derivative: E1 Fourier inversion versus pathwise quasi-Monte Carlo. The final column is the RMSE of the full \(\gamma=1\) state-linear price minus its first-order approximation.}
  \label{tab:price_derivative_results}
  \small
  \begin{tabular}{lrrrr}
\toprule
Swaption & Derivative RMSE & Maximum error & ATM vol equivalent & $\gamma=1$ remainder RMSE \\
\midrule
$\mathrm{3Yx7Y}$ & $7.80e-07$ & $1.17e-06$ & 0.0028 bp & $2.90e-07$ \\\\
$\mathrm{10Yx10Y}$ & $2.15e-06$ & $3.25e-06$ & 0.0103 bp & $1.57e-06$ \\\\
\bottomrule
\end{tabular}

\end{table}

\begin{figure}[H]
  \centering
  \includegraphics[width=0.78\textwidth]{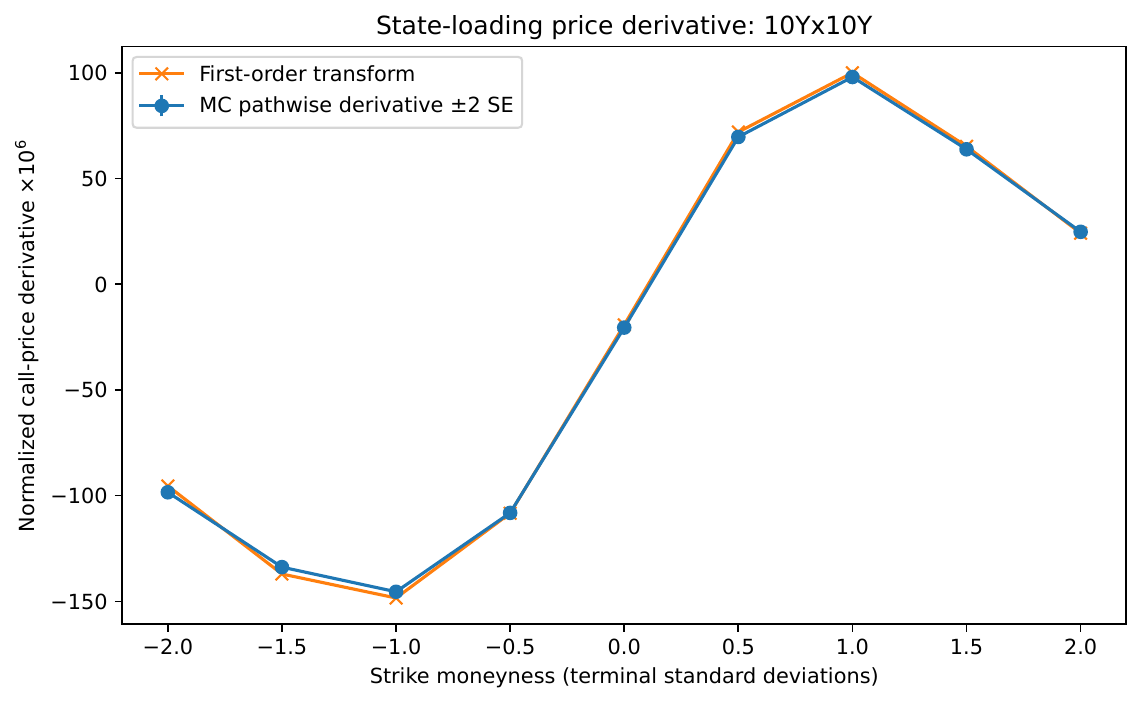}
  \caption{Pathwise and Fourier first-order payer-swaption price derivatives for the \(10\mathrm{Y}\times10\mathrm{Y}\) case. The two estimates are nearly indistinguishable across the nine-strike grid.}
  \label{fig:price_derivative_10y}
\end{figure}

The first-order truncation in the auxiliary parameter can be checked without a transform approximation because \cref{eq:pathwise_affine_gamma} holds pathwise. For a \(3\mathrm{Y}\times7\mathrm{Y}\) swaption at \(\phi=-0.5+15\ii\), define
\begin{equation}
  \mathcal R(\gamma)
  =\left|M(\gamma)-M(0)-\gamma M^{(1)}(0)\right|.
  \label{eq:gamma_remainder}
\end{equation}
Across the tested loading-scale values, \(\mathcal R(\gamma)/\gamma^2\) remains nearly constant. This is the expected second-order scaling of the first-order truncation remainder and verifies the order of the omitted terms.

\begin{figure}[H]
  \centering
  \includegraphics[width=0.72\textwidth]{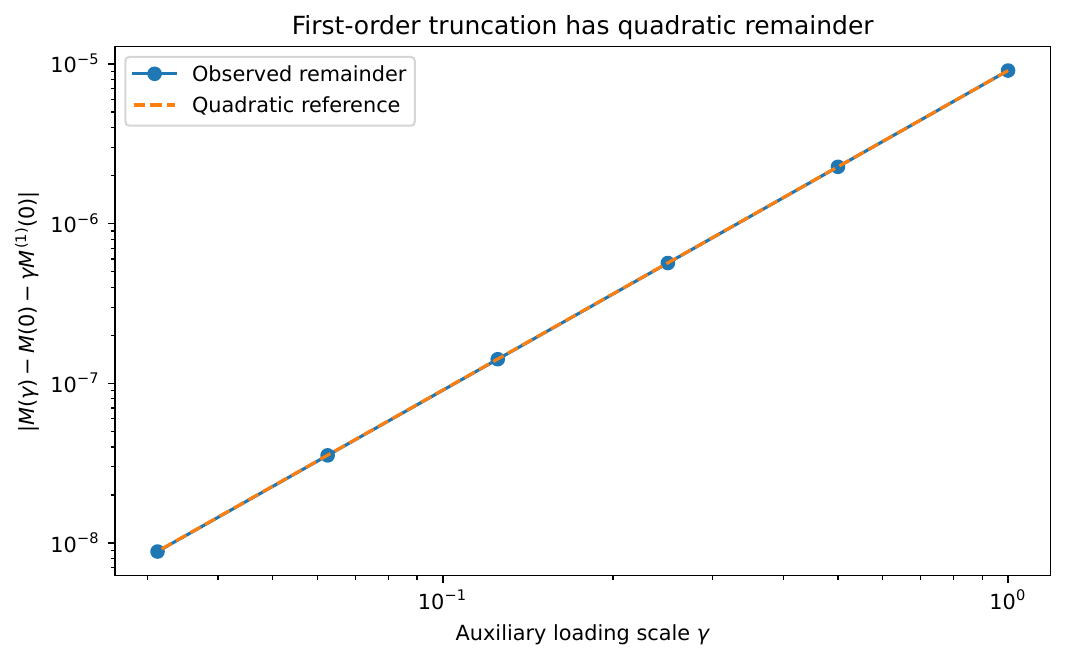}
  \caption{Quadratic scaling of the transform remainder in the auxiliary loading parameter.}
  \label{fig:gamma_scaling}
\end{figure}

\subsection{E1 transform stress and limitations of the finite E1 representation}
\label{sec:stress_results}

The first-variation reduction and its finite E1 implementation involve distinct approximation steps. Conditional on the prescribed centering path and frozen annuity-measure drift, \cref{prop:first_variation} establishes exact affine closure for the first variation within the auxiliary state-linear loading family, with coefficient functions satisfying one-dimensional equations in volatility. This result does not require a polynomial approximation in the volatility state. The full-scale first-order truncation remains separate, and the finite E1 implementation adds a polynomial representation in \(v=\sigma-\theta\). Its numerical accuracy can therefore deteriorate when the volatility process explores a sufficiently wide range around the expansion point \(v=0\).

The E1 transform stress isolates the finite E1 representation error from the underlying one-dimensional first-variation equations. 
For the \(5\mathrm{Y}\times10\mathrm{Y}\) swaption, the degree-six E1 correction differs from the one-dimensional PDE solution by \(105.4\%\), \(61.0\%\), \(19.1\%\), and \(1.2\%\) of the PDE correction at Fourier frequencies \(5\), \(15\), \(30\), and \(60\), respectively. Increasing the polynomial degree does not lead to convergence. At Fourier frequency \(15\),
\begin{equation}
  \begin{array}{c|ccccc}
    N & 3 & 4 & 6 & 8 & 10 \\
    \hline
    |P_N-P_{\mathrm{PDE}}|
      &1.66\times10^{-3}
      &1.31\times10^{-3}
      &5.04\times10^{-3}
      &2.75\times10^{-2}
      &4.21\times10^{-1}
  \end{array}.
  \label{eq:stress_degree_table}
\end{equation}

\begin{figure}[H]
  \centering
  \includegraphics[width=0.76\textwidth]{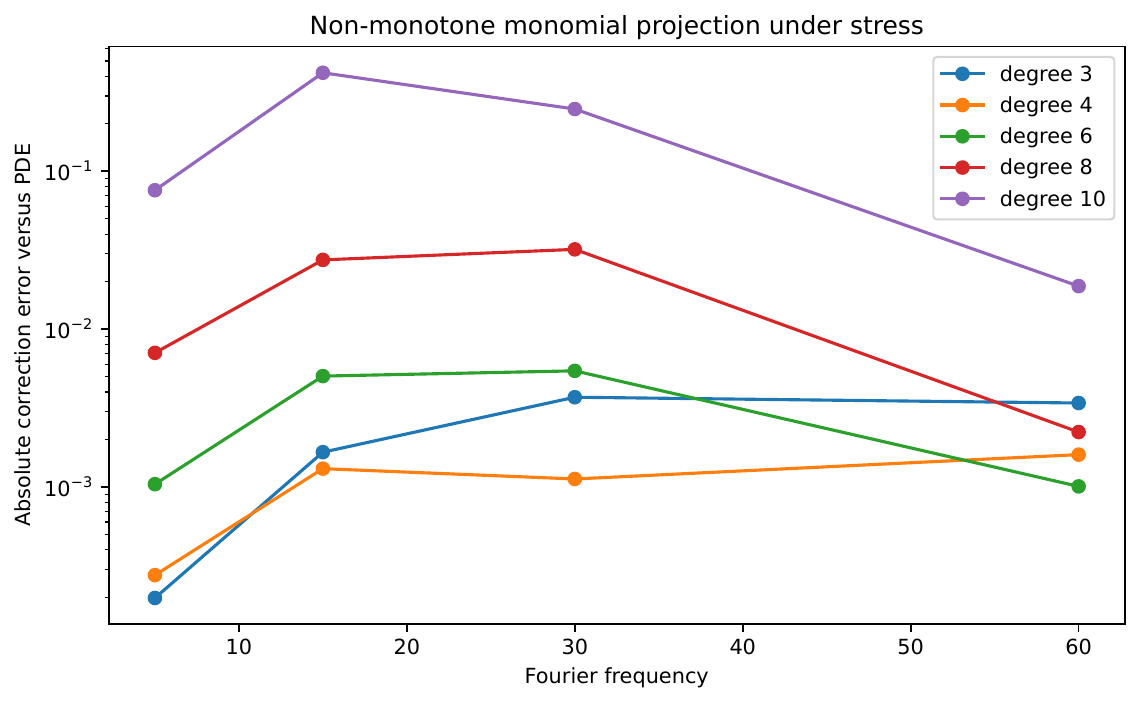}
  \caption{Sensitivity of the finite monomial representation to polynomial degree in the E1 transform stress. The error relative to the one-dimensional PDE is non-monotone and eventually increases rapidly with degree.}
  \label{fig:stress_degree_sensitivity}
\end{figure}

The instability is consistent with the range of volatility states encountered in this specification. The simulated terminal-volatility quantiles at the \(0.1\%\), \(1\%\), \(5\%\), \(50\%\), \(95\%\), \(99\%\), and \(99.9\%\) levels are approximately
\begin{equation}
  0.16,\quad0.23,\quad0.33,\quad0.82,\quad1.97,\quad2.71,\quad3.68.
  \label{eq:stress_sigma_quantiles}
\end{equation}
With \(\theta=1\), this corresponds to a wide range of \(v=\sigma-\theta\), over which a monomial expansion about \(v=0\) is poorly conditioned. Higher powers of \(v\) are consequently no longer uniformly controlled, and increasing the truncation degree need not improve the approximation.

The distinct Broad-volatility stress in \cref{tab:calibration_regimes} shows the same numerical instability at calibration level. Under the finite E1 representation, the held-out RMSE increases from \(\BroadEoneFrozenHeldoutBp\) to \(\BroadEoneCorrectedHeldoutBp\) normal-volatility bp after introducing the first-order correction, and degree and time refinement can change individual prices by as much as \(\BroadEoneRefinementBp\) bp. This demonstrates the failure of the finite monomial representation for this broad-volatility specification.

To evaluate the first-order equations without the monomial approximation, we solve the same one-dimensional system by Chebyshev--Lobatto collocation in \(\ell=\log\sigma\) \citep{Trefethen2000}. The computational domain is chosen from simulated log-volatility quantiles with a fixed extension beyond the sampled range. Convergence is assessed with respect to polynomial degree, domain width, time discretization, Fourier truncation, and quadrature. The final quadrature refinement changes corrected normal-volatility values by at most \(\BroadSpectralFinalQuadratureBp\) bp. The resulting call prices are positive, decreasing in strike, and convex, and the transform and option prices agree with the corresponding Monte Carlo calculations within their sampling uncertainty.

The finite-difference PDE and Chebyshev solutions also agree within the prescribed deterministic numerical tolerance on the broad-volatility calibration panel. The remaining discrepancy with the fully nonlinear model can therefore be attributed to the first-order model approximation rather than to the numerical solution of the one-dimensional equations. Its maximum quote-level total first-order residual is \(\BroadStressResidualBp\) normal-volatility bp, exceeding the tolerance of \(\QuoteResidualThresholdBp\) bp defined in \cref{sec:validity_regions}. Thus, in this stress specification, the affine first-variation reduction within the auxiliary state-linear family remains applicable and its one-dimensional equations can be solved accurately, while the finite E1 representation is numerically unreliable and the first-order approximation at \(\gamma=1\) does not attain the required quote-level accuracy.

\section{Discussion and limitations}\label{sec:discussion}

\subsection{Interpretation of the calibration results}

The calibration results indicate that freezing the swap-rate loading introduces a systematic source of parameter bias. When the frozen specification is calibrated to prices generated by the nonlinear model, the stochastic volatility parameters partially compensate for the omitted state dependence of the swap-rate sensitivity. Restoring this dependence to first order substantially improves the recovery of \(\beta\), \(\epsilon\), \(\nu\), and \(\rho\), while also reducing pricing errors at held-out strikes. These improvements require no additional calibration parameters and remain present in the six-parameter Two buckets specification.

Furthermore, when the fitted parameter vectors are evaluated under the same nonlinear dynamics, the estimates obtained from the corrected calibration reproduce the data-generating swaption surface more closely on both the fitting and held-out panels. The parameter shifts induced by the frozen approximation therefore have consequences for the nonlinear price map itself, supporting the interpretation that the frozen calibration partly compensates for approximation error through the stochastic volatility parameters. The singular-value analysis indicates that the improvement arises primarily from reduced approximation bias rather than from improved local calibration identifiability.
\subsection{Interpretation of the transform results}

Conditional on the prescribed centering path and frozen annuity-measure drift, the first variation within the auxiliary state-linear loading family reduces to one-dimensional equations in volatility. This affine closure does not make the nonlinear loading, the evaluation at \(\gamma=1\), or the centering path exact. Combining the equations with the E1 approximation for quadratic-drift lognormal stochastic volatility gives the finite ODE system. In the ordinary regimes, its approximation error is small relative to the state-loading correction and to the calibration effects of freezing the loading.

The E1 transform stress and the separate Broad-volatility calibration stress distinguish the first-variation reduction from its finite E1 implementation. When the volatility distribution spans a sufficiently wide range around the expansion point, the low-degree monomial representation in \(v=\sigma-\theta\) becomes numerically unreliable and increasing the polynomial degree need not improve it. The same one-dimensional equations can instead be solved directly in log volatility with the PDE/spectral representation. This extends the numerical reach of the first-variation equations beyond E1, although the first-order approximation at full loading scale can itself fail in more extreme parameter regimes.

\subsection{Remaining approximations and possible extensions}

Several approximations from the factor HJM swaption construction remain in the corrected transform.

First, the annuity-measure coefficient \(L_X(t,X,Y)\) is still evaluated along the deterministic expected-state path. Across the five parameter regimes considered here, the common-random-number decomposition attributes between \(0.02\%\) and \(2.09\%\) of the state-loading price effect to this approximation, less than the corresponding loading-linearization error in each case. Linearizing \(L_X\) would introduce additional first-order source terms and is left for future work.

Second, the nonlinear swap-rate loading is approximated only to first order in the centered rate states. This approximation is accurate in the ordinary regimes, but the stress experiments show that the quadratic and higher-order state dependence can become non-negligible. Extending the perturbation to second order would enlarge the invariant state basis to
\begin{equation}
  1,\quad \xi,\quad \eta,\quad \xi^2,\quad \xi\eta,\quad \eta^2,
\end{equation}
with a corresponding increase in the number of one-dimensional volatility equations. Such an extension would be most relevant in parameter regimes for which the measured loading-linearization error is no longer small.

Third, the corrected transform contains only the first-order term in the loading parameter \(\gamma\). In the ordinary-regime experiments, the omitted remainder is small and scales quadratically as predicted by the perturbation expansion. This observation is specific to the parameter regimes considered and is not a uniform error bound over the model parameter space.

Finally, both the E1 and direct log-volatility PDE/spectral implementations solve the first-order approximation rather than the fully nonlinear swaption problem, so we compare them with independently generated nonlinear-model prices. In the ordinary regimes, the remaining discrepancy is small relative to the state-loading correction. In the two stress regimes, the one-dimensional first-variation equations still converge numerically, but the maximum quote-level discrepancy exceeds the tolerance in \cref{sec:validity_regions}. Numerical convergence of these equations can therefore persist after higher-order loading effects have become material.

\subsection{Scope of the empirical study}
The numerical study uses synthetic one-factor specifications to isolate the effect of state dependence in the swap-rate loading. It establishes parameter-recovery and pricing improvements in these controlled settings, not empirical performance on observed swaption markets. Future work should compare the frozen and corrected specifications on historical swaption cubes and test the stability of calibrated parameter term structures across dates and market conditions.

\section{Conclusion}\label{sec:conclusion}

Freezing the swap-rate loading in the factor HJM stochastic volatility construction removes the dependence of conditional swap-rate variance on the current yield-curve state. The synthetic calibration experiments show that this approximation can materially affect the interpretation of the stochastic volatility parameters: when the frozen model is calibrated to prices generated under the nonlinear loading, the fitted volatility parameters partially compensate for the omitted state dependence. Introducing the first-order Taylor correction to the loading substantially reduces this parameter bias and improves pricing at held-out strikes, without adding calibration parameters. Repricing both fitted parameter vectors under the full nonlinear model confirms that the corrected estimates reproduce the nonlinear vanilla price map more accurately in the considered regimes.

For the scalar Cheyette specification with quadratic-drift lognormal stochastic volatility, conditional on the prescribed centering path and frozen annuity-measure drift, the first variation within the auxiliary state-linear loading family admits an exact affine representation in the centered rate states. The associated coefficient functions satisfy one-dimensional parabolic equations in the volatility state. This statement does not remove the loading Taylor remainder or the full-scale \(\gamma=1\) truncation. Combining the equations with the existing E1 approximation yields a finite-dimensional block-linear ODE system that can be used within the same Fourier-pricing framework as the original frozen-loading model.

In the ordinary regimes, the first-order residual is small relative to the state-loading effect, and the correction materially improves parameter recovery. The singular-value analysis attributes this improvement primarily to reduced approximation bias rather than to a substantial change in local identifiability. Under broader volatility dynamics, the finite E1 representation can become unreliable because the monomial expansion in \(v=\sigma-\theta\) is unstable. The one-dimensional first-variation equations can still be solved accurately in log volatility, although the first-order approximation itself may fail at full loading scale.

In the ordinary regimes considered, the state-linear loading materially reduces calibration bias without adding parameters. Its accuracy must nevertheless be checked when the rate or volatility state distribution moves beyond those regimes. Future work should test the correction in multifactor models and on historical swaption surfaces.

\section*{Code availability}

Code and numerical outputs for this paper are available at
\url{https://github.com/quaere-verum/FHJM-SV-Calibration}.

\appendix

\section{Technical details for the first variation}\label{app:first_variation_details}

\subsection{Calendar time and time to expiry}\label{app:phase1_time}

Let \(U_\gamma(t,x)=\E_{t,x}[\Psi(X_T^\gamma)]\), where the time-inhomogeneous
Markov generator in calendar time is \(\cL_{\gamma,t}\).  The backward
Kolmogorov equation and terminal condition are
\begin{equation}
  \partial_t U_\gamma(t,x)+\cL_{\gamma,t}U_\gamma(t,x)=0,
  \qquad U_\gamma(T,x)=\Psi(x).
  \label{eq:calendar_backward}
\end{equation}
Define \(\tau=T-t\), \(\widehat U_\gamma(\tau,x)=U_\gamma(T-\tau,x)\), and,
for every deterministic coefficient path \(c\),
\begin{equation}
  \overleftarrow c(\tau):=c(T-\tau),
  \qquad
  \overleftarrow{\cL}_{\gamma,\tau}:=\cL_{\gamma,T-\tau}.
\end{equation}
The chain rule gives
\begin{equation}
  \partial_\tau\widehat U_\gamma
  =-\partial_tU_\gamma(T-\tau)
  =\cL_{\gamma,T-\tau}\widehat U_\gamma
  =\overleftarrow{\cL}_{\gamma,\tau}\widehat U_\gamma,
  \qquad \widehat U_\gamma(0)=\Psi.
  \label{eq:tau_backward}
\end{equation}
All time-to-expiry equations in the main text use the reversed coefficients
in \cref{eq:tau_backward}.

\subsection{Differentiability and Feynman--Kac representation}
\label{app:phase1_differentiability}

Under the frozen annuity-measure drift the state process is independent of
\(\gamma\).  Let
\begin{equation}
  \Delta S^0_{t,T}:=S_T^0-S_t^0,
  \qquad
  \mathcal I_{t,T}:=\int_t^T\sigma_s h(s,Z_s)^{\top}\dd W_s.
\end{equation}
Then, path by path,
\begin{equation}
  S_T^\gamma-S_t^\gamma=\Delta S^0_{t,T}+\gamma\mathcal I_{t,T}.
  \label{eq:pathwise_affine_rigorous}
\end{equation}

\begin{theorem}[Differentiability of the transform]
\label{thm:transform_differentiability}
Fix a finite horizon and a real interval \(I=[-G,G]\) of auxiliary loading-scale values.
Assume that the frozen state SDE has a nonexplosive strong solution and that
the stochastic integrals defining \(\Delta S^0_{t,T}\) and
\(\mathcal I_{t,T}\) exist. For \(\phi=\ii u\), \(u\in\R\), assume
\(\mathcal I_{t,T}\in L^1\) for one derivative and
\(\mathcal I_{t,T}\in L^2\) for two. Then
\(M_\gamma(\ii u)=\E[e^{-\ii u(\Delta S^0_{t,T}+\gamma\mathcal I_{t,T})}]\)
is respectively once or
twice continuously differentiable on \(I\), with
\begin{align}
  \partial_\gamma M_\gamma(\phi)
  &=-\phi\E[\mathcal I_{t,T}
      e^{-\phi(\Delta S^0_{t,T}+\gamma\mathcal I_{t,T})}],
  \label{eq:first_derivative_rigorous}\\
  \partial_{\gamma\gamma}M_\gamma(\phi)
  &=\phi^2\E[\mathcal I_{t,T}^2
      e^{-\phi(\Delta S^0_{t,T}+\gamma\mathcal I_{t,T})}].
  \label{eq:second_derivative_rigorous}
\end{align}
In particular, we write
\(M^{(1)}(0):=\left.\partial_\gamma M_\gamma\right|_{\gamma=0}\);
no prime is used for both the family index and its evaluation point.
For complex \(\phi=a+\ii u\), the same conclusion holds on the explicit
domain
\begin{equation}
  \mathcal D_{k,G}
  :=\left\{\phi\in\mathbb{C}:
  \E\!\left[|\mathcal I_{t,T}|^k
    e^{|\Re\phi|(|\Delta S^0_{t,T}|+G|\mathcal I_{t,T}|)}\right]<\infty\right\},
  \qquad k=1,2,
  \label{eq:complex_transform_domain}
\end{equation}
with \(k\) chosen for the required derivative.
\end{theorem}

\begin{proof}
For characteristic arguments the absolute values of the first and second
pointwise derivatives are bounded by \(|u\mathcal I_{t,T}|\) and
\(u^2\mathcal I_{t,T}^2\), respectively.
For complex arguments and \(|\gamma|\leq G\), they are bounded by
\(|\phi||\mathcal I_{t,T}|
 e^{|a|(|\Delta S^0_{t,T}|+G|\mathcal I_{t,T}|)}\) and
\(|\phi|^2\mathcal I_{t,T}^2
 e^{|a|(|\Delta S^0_{t,T}|+G|\mathcal I_{t,T}|)}\). The stated assumptions therefore give
integrable dominators.  Dominated convergence, applied first to difference
quotients and then to the first derivative, proves
\cref{eq:first_derivative_rigorous,eq:second_derivative_rigorous} and their
continuity in \(\gamma\).
\end{proof}

These expectations define a Feynman--Kac \emph{mild} solution through the
time-inhomogeneous Markov evolution family; Duhamel's formula gives the mild
inhomogeneous equation for the first variation.  A classical solution follows
under, for example, sufficient coefficient and terminal-data regularity plus
interior hypoelliptic smoothing (a H\"ormander bracket condition) and the
corresponding growth estimates, or from a separately established classical
well-posedness result for the one-dimensional volatility problems.  Uniform
ellipticity is neither assumed nor generally true: the Cheyette auxiliary
\(Y\) state has no direct diffusion.

\subsection{Multidimensional extension}
\label{app:phase1_multidimensional}

Let \(Z_t\in\R^n\), \(W_t\in\R^m\), and let \(B_t\in\R^r\) be independent of
\(W\).  Let \(K_t\in\R^{n\times n}\), \(r_t(v)\in\R^n\),
\(G_t\in\R^{n\times m}\), \(\beta_t\in\R^m\),
\(\epsilon_t\in\R^r\), \(a_0(t)\in\R^m\), and
\(J_t\in\R^{m\times n}\). Consider
\begin{align}
  \dd Z_t&=(K_tZ_t+r_t(v_t))\dd t+\sigma(v_t)G_t\dd W_t,
  \label{eq:multi_Z}\\
  \dd v_t&=b_v(t,v_t)\dd t
    +\sigma(v_t)(\beta_t^\top\dd W_t+\epsilon_t^\top\dd B_t),
  \label{eq:multi_v}\\
  \dd S_t^\gamma&=\sigma(v_t)
    [a_0(t)+\gamma J_tZ_t]^\top\dd W_t.
  \label{eq:multi_S}
\end{align}
Thus \(J_t\) maps rate state into Brownian-loading space. Put
\(\nu_t^2=\|\beta_t\|^2+\|\epsilon_t\|^2\), \(\sigma=\sigma(v)\), and define
the scalar frozen operator
\begin{equation}
  \cA_\phi q
  =\tfrac12\nu_t^2\sigma^2q_{vv}
   +(b_v-\phi\sigma^2a_0^\top\beta_t)q_v
   +\tfrac12\phi^2\sigma^2\|a_0\|^2q.
  \label{eq:multi_A}
\end{equation}

\begin{proposition}[Multidimensional first variation]
\label{prop:multidimensional_first_variation}
Under the assumptions of \cref{thm:transform_differentiability}, and assuming
uniqueness in the relevant mild-solution growth class, conditional on the
prescribed centering path and frozen drift, the first transform variation
within the auxiliary state-linear family has the exact affine form
\begin{equation}
  f_1(t,z,v)=g_0(t,v)+g(t,v)^\top z,
  \qquad g(t,v)\in\mathbb{C}^n.
  \label{eq:multi_affine_form}
\end{equation}
In reversed time its non-conjugated coefficient system is
\begin{align}
  \partial_\tau g
  &=\cA_\phi g+K^\top g
    +\sigma^2J^\top(\phi^2a_0f_0-\phi\beta f_{0,v}),
  \label{eq:multi_direct_g}\\
  \partial_\tau g_0
  &=\cA_\phi g_0+r^\top g
    +\sigma^2(G\beta)^\top g_v
    -\phi\sigma^2(Ga_0)^\top g,
  \label{eq:multi_direct_g0}
\end{align}
with zero initial data.  Here \(\cA_\phi\) acts componentwise on \(g\).

If \(f_0\ne0\), write \(f_0=e^A\), \(\psi=A_v\),
\(f_1=f_0(p_0+p^\top z)\), and
\begin{equation}
  \cD_\phi q
  =\tfrac12\nu^2\sigma^2q_{vv}
   +(b_v+\nu^2\sigma^2\psi-\phi\sigma^2a_0^\top\beta)q_v.
\end{equation}
The equivalent log-conjugated system is
\begin{align}
  \partial_\tau p
  &=\cD_\phi p+K^\top p
    +\sigma^2J^\top(\phi^2a_0-\phi\beta\psi),
  \label{eq:multi_log_p}\\
  \partial_\tau p_0
  &=\cD_\phi p_0
    +(r+\sigma^2G\beta\psi-\phi\sigma^2Ga_0)^\top p
    +\sigma^2(G\beta)^\top p_v.
  \label{eq:multi_log_p0}
\end{align}
\end{proposition}

\begin{proof}
For the flow argument, fix \((v,W,B)\) after the initial time.  Variation of
constants in \cref{eq:multi_Z} gives
\(Z_s=\Phi_{s,t}z+\zeta_{s,t}\), where
\(\partial_s\Phi_{s,t}=K_s\Phi_{s,t}\) and \(\Phi_{t,t}=I_n\).  Hence the
pathwise first loading integral is affine in its initial state:
\(\mathcal I_{t,T}=\mathcal I^0_{t,T}+q_{t,T}^\top z\). The frozen exponential weight is
independent of \(z\), so
\(-\phi\E[\mathcal I_{t,T}e^{-\phi\Delta S^0_{t,T}}\mid Z_t=z,v_t=v]\)
is affine in \(z\).

For the generator argument, the frozen generator maps the vector space
\(\{q_0(v)+q(v)^\top z\}\) into itself.  Its affine drift contributes
\(K^\top q\); the \(Z\)--\(v\) covariance contributes
\(\sigma^2(G\beta)^\top q_v\); and the frozen \(S\)--\(Z\) covariance
contributes \(-\phi\sigma^2(Ga_0)^\top q\).  Differentiating the loading gives
the source
\(z^\top\sigma^2J^\top(\phi^2a_0f_0-\phi\beta f_{0,v})\).
Collection of constant and linear terms proves
\cref{eq:multi_direct_g,eq:multi_direct_g0}; division by \(f_0\) and the
product rule give \cref{eq:multi_log_p,eq:multi_log_p0}.
\end{proof}

For scalar Cheyette take \(n=2\), \(m=1\),
\(Z=(\xi,\eta)^\top\),
\(K=\left(\begin{smallmatrix}-\lambda&1\\0&-2\lambda\end{smallmatrix}\right)\),
\(G=(\alpha,0)^\top\), and \(J=(a_x,a_y)\). The proposition then reduces
exactly to \cref{eq:g_system,eq:p_system}; in particular the \(+g_x\) and
\(+p_x\) couplings are entries of \(K^\top\).

\subsection{Deterministic-volatility Gaussian limit}
\label{app:phase1_gaussian_limit}

As an independent sign and cross-variation check, let
\begin{equation}
  \dd X_t=-\lambda X_t\dd t+c\dd W_t,
  \qquad
  \dd S_t^\gamma=(b+\gamma d\,X_t)\dd W_t,
  \qquad X_0=0,
\end{equation}
with constant deterministic \(b,c,d\). For real $\phi$, under the normalized exponential
tilt generated by \(e^{-\phi bW_T}\),
\(\widetilde W_t=W_t+\phi bt\) is Brownian and
\begin{equation}
  \widetilde\E[X_t]
  =-\phi bc\frac{1-e^{-\lambda t}}{\lambda}.
\end{equation}
Writing the original stochastic integral in terms of \(\widetilde W\) gives
\begin{equation}
  \widetilde\E\!\left[\int_0^T d\,X_t\dd W_t\right]
  =\phi^2dcb^2
    \left[\frac{T}{\lambda}
      -\frac{1-e^{-\lambda T}}{\lambda^2}\right].
\end{equation}
Consequently
\begin{equation}
  \boxed{
  \left.\partial_\gamma\log M_\gamma(\phi)\right|_0
  =-\phi^3dcb^2
    \left[\frac{T}{\lambda}
      -\frac{1-e^{-\lambda T}}{\lambda^2}\right].}
  \label{eq:gaussian_closed_form}
\end{equation}
The resulting identity extends to complex $\phi$ on the common transform
domain by analyticity. The bracket tends continuously to \(T^2/2\) as \(\lambda\to0\).  For
time-dependent deterministic coefficients, let
\(J'(t)=c(t)b(t)-\lambda J(t)\), \(J(0)=0\).  The correction becomes
\begin{equation}
  -\phi^3\int_0^T d(t)b(t)J(t)\dd t.
  \label{eq:gaussian_time_dependent}
\end{equation}
\Cref{eq:gaussian_closed_form} fixes the Fourier sign and vanishes
when the common-noise coefficient \(c\), hence the relevant cross variation,
vanishes. Both formulas agree with numerical solutions of the reference ODE,
including \(\lambda=0\) and asymmetric calendar-time paths.

\section{Baseline E1 equations}\label{app:E1}

For completeness, this appendix records the quadratic log-transform equations used as the baseline finite-dimensional approximation. Let
\begin{equation}
  v=\sigma-\theta,
  \qquad
  b_v=k_0-k_1v-k_2v^2,
  \qquad
  \sigma=\theta+v,
\end{equation}
and consider the frozen-loading transform
\begin{equation}
  f_0(\tau,v;\phi)=\exp\!\left(A_0(\tau)+A_1(\tau)v+A_2(\tau)v^2\right).
\end{equation}
Its log-PDE is
\begin{equation}
  \partial_\tau A
  =b_vA_v
   +\frac12\nu^2\sigma^2\left(A_{vv}+A_v^2\right)
   -\phi\beta a_0\sigma^2A_v
   +\frac12\phi^2a_0^2\sigma^2.
  \label{eq:app_log_pde}
\end{equation}
Substitute \(A=A_0+A_1v+A_2v^2\), expand the right-hand side in powers of \(v\), and truncate at degree two. Writing
\begin{equation}
  \dot A_j=\mathcal F_j(A_1,A_2;t,\phi),
  \qquad A_j(0)=0,
\end{equation}
we obtain the E1 equations. A compact and less error-prone representation is coefficient extraction from
\begin{equation}
  \mathcal P(v)
  =(k_0-k_1v-k_2v^2)(A_1+2A_2v)
  +\frac12\nu^2(\theta+v)^2
     \left[2A_2+(A_1+2A_2v)^2\right]
\end{equation}
\begin{equation}
  \hspace{3.5em}
  -\phi\beta a_0(\theta+v)^2(A_1+2A_2v)
  +\frac12\phi^2a_0^2(\theta+v)^2.
  \label{eq:app_polynomial_rhs}
\end{equation}
Then
\begin{equation}
  \boxed{
  \dot A_0=[v^0]\mathcal P(v),\qquad
  \dot A_1=[v^1]\mathcal P(v),\qquad
  \dot A_2=[v^2]\mathcal P(v),
  }
  \label{eq:app_E1_coeff_extraction}
\end{equation}
where \([v^j]\) denotes coefficient extraction. The discarded \(v^3\) and \(v^4\) coefficients provide a baseline residual diagnostic. This formulation is algebraically equivalent to the E1 equations used in \citet{SeppRakhmonovLogSV2023,SeppRakhmonovFHJM2025} under the transform convention of \cref{eq:transform_definition}.

\section{Pathwise first-variation identities}\label{app:mc_identities}

Under the frozen annuity drift, \((\xi,\eta,\sigma)\) does not depend on the auxiliary parameter \(\gamma\). Hence
\begin{equation}
  \Delta S^\gamma_{0,T}
  =\int_0^T\sigma_t(a_0+\gamma h_t)\dd W_t
  =\Delta S^0_{0,T}+\gamma\mathcal I_{0,T},
  \qquad
  \mathcal I_{0,T}=\int_0^T\sigma_th_t\dd W_t.
\end{equation}
For any complex \(\phi\) for which differentiation and expectation can be interchanged,
\begin{equation}
  \frac{\partial}{\partial\gamma}
  \E[e^{-\phi\Delta S^\gamma_{0,T}}]
  =\E[-\phi\mathcal I_{0,T}
    e^{-\phi(\Delta S^0_{0,T}+\gamma\mathcal I_{0,T})}].
\end{equation}
At \(\gamma=0\), this gives \cref{eq:pathwise_mgf_derivative}. A second differentiation gives
\begin{equation}
  \frac{\partial^2}{\partial\gamma^2}
  \E[e^{-\phi\Delta S^\gamma_{0,T}}]
  =\E[\phi^2\mathcal I_{0,T}^2
    e^{-\phi(\Delta S^0_{0,T}+\gamma\mathcal I_{0,T})}],
\end{equation}
which explains the quadratic remainder in \cref{eq:gamma_remainder} under the required integrability conditions.

For the call payoff \(g(z)=(F_0+z-K)^+\), the pathwise derivative is
\begin{equation}
  \frac{\partial}{\partial\gamma}
  g(\Delta S^0_{0,T}+\gamma\mathcal I_{0,T})
  =\mathcal I_{0,T}\mathbf 1_{\{F_0+\Delta S^0_{0,T}
    +\gamma\mathcal I_{0,T}>K\}}
\end{equation}
almost surely. Provided \(\mathcal I_{0,T}\in L^1\) and the terminal distribution has no atom at the strike, taking expectations at \(\gamma=0\) yields \cref{eq:pathwise_call_derivative}.

\section{Additional parameter-recovery results}
\label{app:additional_calibration}

\Cref{tab:full_parameter_recovery} reports the fitted values of
\(\alpha\), \(\beta\), and \(\epsilon\) for the Moderate and Negative skew
one-bucket specifications. The rate-volatility scale parameter \(\alpha\) is
recovered accurately under both the frozen and state-linear specifications.
The more pronounced effect of the frozen-loading approximation appears in
\(\beta\), and consequently in the implied volatility magnitude
\(\nu=\sqrt{\beta^2+\epsilon^2}\) and instantaneous rate--volatility
correlation \(\rho=\beta/\nu\). Incorporating the state-linear loading
substantially reduces these parameter-recovery errors. Results for the stressed
specifications, for which the log-volatility PDE/spectral representation is
used, are discussed separately in
\cref{sec:coefficient_calibration_results}.

\begin{table}[H]
  \centering
  \caption{Parameter recovery for the Moderate and Negative skew one-bucket
  specifications under the frozen and state-linear calibrations.}
  \label{tab:full_parameter_recovery}
  \small
  \resizebox{\textwidth}{!}{\begin{tabular}{llrrrrrrrr}
\toprule
Regime & Pricer & $\alpha$ & $\beta$ & $\epsilon$ & $\nu$ & $\rho$ & Fit & Holdout \\
\midrule
Moderate & Truth & 0.012000 & -0.050000 & 0.180000 & 0.186815 & -0.267644 & -- & -- \\
Moderate & Frozen E1 & 0.012009 & -0.019950 & 0.162668 & 0.163887 & -0.121730 & 0.3771 & 0.3169 \\
Moderate & Corrected E1 & 0.012002 & -0.049937 & 0.178255 & 0.185117 & -0.269760 & 0.0151 & 0.0118 \\
\addlinespace
Negative skew & Truth & 0.012000 & -0.200000 & 0.300000 & 0.360555 & -0.554700 & -- & -- \\
Negative skew & Frozen E1 & 0.011980 & -0.172068 & 0.294004 & 0.340655 & -0.505109 & 0.3136 & 0.2622 \\
Negative skew & Corrected E1 & 0.012000 & -0.196575 & 0.294946 & 0.354450 & -0.554592 & 0.0381 & 0.0323 \\
\bottomrule
\end{tabular}
}
\end{table}

\bibliographystyle{plainnat}
\bibliography{references}

\end{document}